\documentclass[11pt]{article}
\usepackage{amsfonts,amsmath,amsthm,amssymb}

\usepackage{comment, enumerate}
\usepackage[margin=1in]{geometry}
\usepackage{soul,color}
\usepackage{mathrsfs}
\usepackage{mathdots}
\usepackage[linesnumbered,boxed,ruled,vlined]{algorithm2e}
\usepackage[capitalise]{cleveref}

\usepackage{bbm}

\usepackage{url}

\usepackage{tabu}

\usepackage{xfrac}

\usepackage[normalem]{ulem}

\newcommand{\poly}{\mathop{\mathrm{poly}}}
\newcommand{\polylog}{\mathop{\mathrm{polylog}}}

\newcommand{\rank}{\mathop{\operatorname{rank}}}
\newcommand{\ch}{\mathop{\operatorname{ch}}}
\newcommand{\nnz}{\mathop{\operatorname{nnz}}}
\newcommand{\nnzr}{\mathop{\operatorname{nnz_r}}}
\newcommand{\nnzc}{\mathop{\operatorname{nnz_c}}}
\newcommand{\argmax}{\mathop{\operatorname{argmax}}}

\newcommand{\F}{\mathbb{F}}

\newcommand{\Z}{\mathbb{Z}}
\newcommand{\N}{\mathbb{N}}

\newcommand{\C}{\mathbb{C}}

\newcommand{\rig}{\mathcal{R}}

\newcommand{\eps}{\varepsilon}

\theoremstyle{plain}\newtheorem{theorem}{Theorem}[section]
\newtheorem{lemma}[theorem]{Lemma}
\newtheorem{proposition}[theorem]{Proposition}
\newtheorem{corollary}[theorem]{Corollary}

\newtheorem{remark}[theorem]{Remark}
\theoremstyle{definition}
\newtheorem{definition}[theorem]{Definition}

\usepackage[utf8]{inputenc}

\title{Kronecker Products, Low-Depth Circuits, and Matrix Rigidity} 

\author{Josh Alman\footnote{Harvard University. \texttt{jalman@seas.harvard.edu}. Supported by a Michael O. Rabin postdoctoral fellowship.}}

\begin{document}

\maketitle

\begin{abstract}
For a matrix $M$ and a positive integer $r$, the rank $r$ rigidity of $M$ is the smallest number of entries of $M$ which one must change to make its rank at most $r$. There are many known applications of rigidity lower bounds to a variety of areas in complexity theory, but fewer known applications of rigidity upper bounds. In this paper, we use rigidity upper bounds to prove new upper bounds in a few different models of computation. Our results include:
\begin{itemize}
    \item For any $d>1$, and over any field $\F$, the $N \times N$ Walsh-Hadamard transform has a depth-$d$ linear circuit of size $O(d \cdot N^{1 + 0.96/d})$. This circumvents a known lower bound of $\Omega(d \cdot N^{1 + 1/d})$ for circuits with bounded coefficients over $\C$~\cite{pudlak2000note}, by using coefficients of magnitude polynomial in $N$. Our construction also generalizes to linear transformations given by a Kronecker power of any fixed $2 \times 2$ matrix.
    \item The $N \times N$ Walsh-Hadamard transform has a linear circuit of size $\leq (1.81 + o(1)) N \log_2 N$, improving on the bound of $\approx 1.88 N \log_2 N$ which one obtains from the standard fast Walsh-Hadamard transform.
    \item A new rigidity upper bound, showing that the following classes of matrices are not rigid enough to prove circuit lower bounds using Valiant's approach:
    \begin{itemize}
        \item for any field $\F$ and any function $f : \{0,1\}^n \to \F$, the matrix $V_f \in \F^{2^n \times 2^n}$ given by, for any $x,y \in \{0,1\}^n$, $V_f[x,y] = f(x \wedge y)$, and
        \item for any field $\F$ and any fixed-size matrices $M_1, \ldots, M_n \in \F^{q \times q}$, the Kronecker product $M_1 \otimes M_2 \otimes \cdots \otimes M_n$.
    \end{itemize}
    This generalizes recent results on non-rigidity, using a simpler approach which avoids needing the polynomial method.
    \item New connections between recursive linear transformations like Fourier and Walsh-Hadamard transforms, and circuits for matrix multiplication.
\end{itemize}
\end{abstract}

\thispagestyle{empty}
\newpage
\setcounter{page}{1}

\section{Introduction}

For a matrix $M$ and a positive integer $r$, the rank $r$ rigidity of $M$, denoted $\rig_M(r)$, is the smallest number of entries of $M$ which one must change to make its rank at most $r$. Matrix rigidity was introduced by L.~Valiant~\cite{valiant1977graph} as a tool for proving low-depth circuit lower bounds. He showed that for any family $\{M_N\}_{N \in \N}$ of matrices with $M_N \in \F^{N \times N}$, if $\rig_{M_N}(O(N/\log\log N)) \geq N^{1+\eps}$ for any fixed $\eps>0$, then the linear transformation which takes as input a vector $x \in \F^N$ and outputs $M_N  x$ cannot be computed by an arithmetic circuit of size $O(N)$ and depth $O(\log N)$. We say $M_N$ is \emph{Valiant-rigid} if it satisfies this rigidity lower bound. It remains a major open problem to prove that any explicit family of matrices\footnote{We say $\{M_N\}_{N \in \N}$ with $M_N \in \F^{N \times N}$ is \emph{explicit} if there is an algorithm which, on input $N$, outputs $M_N$ in $\poly(N)$ deterministic time.} cannot be computed by circuits of size $O(N)$ and depth $O(\log N)$, and one of the most-studied approaches to this problem is to try to construct an explicit family of Valiant-rigid matrices.

Many researchers have subsequently shown that rigidity lower bounds for explicit matrices, both in this parameter regime and others, would lead to new lower bounds in a variety of areas, including in arithmetic complexity, communication complexity, Boolean circuit complexity, and cryptography. We refer the reader to~\cite{lokam2009complexity} for more on the background and known applications of matrix rigidity. However, despite 40+ years of efforts, and plenty of known applications, there are no known fully explicit constructions of rigid matrices.

A recent line of work~\cite{alman2017probabilistic,dvir2017matrix,dvir2019fourier} has instead shown that a number of families of explicit matrices are in fact not Valiant rigid, including the Walsh-Hadamard transform~\cite{alman2017probabilistic} and the discrete Fourier transform~\cite{dvir2019fourier}. These had been some of the most-studied candidate rigid matrices, which are now ruled out for proving lower bounds using this approach. This raises the question: Do these rigidity upper bounds imply any other interesting upper bounds? Although there are many results showing that rigid matrices imply a variety of lower bounds, there are few known connections showing that rigidity upper bounds would yield new algorithms or circuits.

In this paper, we give new upper bounds in a few different models which make use of recent rigidity upper bounds. Some of them apply rigidity upper bounds directly, while others are inspired by the proof techniques of recent rigidity upper bounds.

\subsection{Low-Depth Linear Circuits} \label{introsec:lowdepth}

We begin by studying linear circuits for computing a linear transformation $M \in \F^{N \times N}$. These are circuits in which the inputs are the $N$ entries of a vector $x \in \F^N$, the outputs must be the $N$ entries of $M x$, and each gate computes an $\F$-linear combination of its inputs. We focus on low-depth circuits with unbounded fan-in gates, so we measure their size by the number of \emph{wires} in the circuit. A special type of linear circuit which we focus on is a \emph{synchronous} linear circuit, in which the inputs to each gate must all have the same depth. One can see that a synchronous linear circuit of size $s$ and depth $d$ for $M$ corresponds to $d$ matrices $M_1, \ldots, M_d$ such that $M = M_1 \times \cdots \times M_d$ and $\nnz(M_1) + \cdots + \nnz(M_d) = s$, where $\nnz(A)$ denotes the number of nonzero entries in matrix $A$. A depth $d$ linear circuit can be converted into a depth $d$ synchronous linear circuit with a multiplicative size blowup of only $d$.

Rigidity upper bounds naturally give depth-$2$ linear circuit constructions. Indeed, it is not hard to see that any $M \in \F^{N \times N}$ has a depth-$2$ linear circuit of size $O(N \cdot \rank(M))$, and a depth-$1$ linear circuit of size $O(\nnz(M))$, and hence, for any $r$, a depth-$2$ linear circuit of size $O(N \cdot r + \rig_M(r))$. Thus, for instance, letting $H_n$ denote the $N \times N$ Walsh-Hadamard transform for $N = 2^n$, using the rigidity upper bound $\rig_{H_n}(N^{1 - \Theta(\eps^2 / \log^2(1/\eps))}) \leq N^{1 + \eps}$ for any $\eps>0$ of~\cite{alman2017probabilistic}, it follows that there is a fixed $\delta>0$ such that $H_n$ has a depth-$2$ linear circuit of size $O(N^{2-\delta})$. 

However, there is actually a smaller and simpler circuit known for $H_n$. Using an approach similar to the fast Walsh-Hadamard Transform, we can see that for any $d$, $H_n$ has a depth-$d$ synchronous linear circuit of size only $O(d \cdot N^{1+1/d})$. (The circuit involves, at each depth, computing $N^{1-1/d}$ independent copies of the $N^{1/d} \times N^{1/d}$ Walsh-Hadamard transform $H_{n/d}$.) Thus, $H_n$ has a depth-$2$ circuit of size only $O(N^{1.5})$, which is much better than $O(N^{2-\delta})$. Despite a fair bit of work by the author, it is unclear how to use the rigidity upper bound of~\cite{alman2017probabilistic} to improve on $O(N^{1.5})$.

Nonetheless, we are able to construct smaller circuits for $H_n$, as well as any other family of transforms defined as the Kronecker power of a fixed matrix, by making use of new, different rigidity upper bounds for $H_n$. For a fixed $2 \times 2$ matrix
$$M = 
\begin{bmatrix}
a & b \\
c & d
\end{bmatrix}$$
over a field $\F$, the family of Kronecker powers of $M$, denoted by $M^{\otimes n} \in \F^{2^n \times 2^n}$, is defined recursively by $M^{\otimes 1} = M$, and for $n \geq 1$, $$M^{\otimes (n+1)} = 
\begin{bmatrix}
a \cdot M^{\otimes n} & b \cdot M^{\otimes n} \\
c \cdot M^{\otimes n} & d \cdot M^{\otimes n}
\end{bmatrix}.$$
For instance, the $2^n \times 2^n$ Walsh-Hadamard transform $H_n$ is defined as $H_n := H_1^{\otimes n}$, where 
$$H_1 := 
\begin{bmatrix}
1 & 1 \\
1 & -1
\end{bmatrix}.$$
Kronecker powers arise naturally in many settings. For instance, when $$M = 
\begin{bmatrix}
1 & 1 \\
1 & \omega
\end{bmatrix}$$ for some element $\omega \in F$, then the linear transformation $M^{\otimes n}$ corresponds to evaluating an $n$-variate multilinear polynomial over $\F$ on all inputs in $\{1,\omega\}^n$.

Our main result is as follows:
\begin{theorem} \label{intro:mainthm}
Let $\F$ be any field, and let $M \in \F^{2 \times 2}$ be any matrix over $\F$. There is a constant $\eps > 0.01526$ such that, for any positive integers $n,d$, the linear transformation $M^{\otimes n} \in \F^{N \times N}$ for $N = 2^n$ has a depth-$d$ synchronous linear circuit of size $2^{\eps} \cdot d \cdot N^{1 + (1-\eps)/d}$. When $M = H_1$, so that $M^{\otimes n}$ is the Walsh-Hadamard transform $H_n$, we can improve the bound to $\eps > 0.04816$. 
\end{theorem}
Our new result shows that $H_n$ has a depth-$2$ linear circuit of size only $O(N^{1.47592})$, and more generally improves the size of a depth-$d$ linear circuit for $H_n$ or any $n$th Kronecker power when $d < o(\log n)$. When $d$ divides $n$, we can improve the upper bound to $d \cdot N^{1 + (1-\eps)/d}$, removing the $2^{\eps}$ factor. This construction may be of practical interest, as it improves on the previous bound of $d \cdot N^{1 + 1/d}$, even for small constant values of $N$ and $d$.

\Cref{intro:mainthm} is also particularly interesting when compared to a lower bound of Pudl{\'a}k~\cite{pudlak2000note} against low-depth linear circuits with \emph{bounded coefficients} for computing $H_n$ over $\C$. Recall that in a linear circuit over $\C$, each gate computes a $\C$-linear combination of its inputs. For a positive real number $c$, we say the circuit has $c$-bounded coefficients if, for each gate, the coefficients of the linear combination are complex numbers of magnitude at most $c$. Motivated by the fact that the best known linear circuits for many important linear transformations, including the Walsh-Hadamard transform and the discrete Fourier transform, use only $1$-bounded coefficients (prior to this paper), a  line of work~\cite{morgenstern1973note,chazelle1994spectral,lokam2001spectral,nisan1996lower,pudlak2000note,burgisser2004lower,raz2002complexity} (see also~\cite[{Section~3.3}]{lokam2009complexity}) has shown strong, often tight lower bounds for linear circuits with bounded coefficients. Pudl{\'a}k~\cite{pudlak2000note} showed that the aforementioned circuit of depth $d$ and size $O(d \cdot N^{1 + 1/d})$ is optimal for bounded coefficient circuits:

\begin{theorem}[\cite{pudlak2000note}] \label{pudlak1}
Any depth $d$ synchronous linear circuit with $c$-bounded coefficients for computing the Walsh-Hadamard transform $H_n \in \C^{N \times N}$ for $N=2^n$ has size $\geq d \cdot N^{1 + 1/d} / c^2$. 
\end{theorem}

Our \Cref{intro:mainthm} circumvents this lower bound by using large coefficients. Indeed, we will see that over $\F = \C$, we use coefficients which are integers of magnitude up to $N^{O(1)}$. That said, it should be noted that, since our coefficients are only $O(\log N)$-bit integers, the additional time required to do the arithmetic for the coefficients of our circuit is still negligible compared to the circuit size savings in any reasonable model of computation.

To our knowledge, this is the first non-trivial upper bound surpassing one of the aforementioned bounded-coefficient lower bounds.
This shows that using larger coefficients can make a substantial difference in the circuit size required, even when computing the linear transformation of a matrix whose entries are all in $\{-1,1\}$. 
At the same time, it is interesting to note that our \Cref{intro:mainthm} works over any field, even a constant-sized finite field like $\F_3$ where there are no `large' coefficients. One could have imagined that overcoming bounded-coefficient lower bounds, when possible, \emph{requires} using an infinite field and large coefficients, but at least in this setting, that is not the case.

Our proof of \Cref{intro:mainthm} begins with a new general framework for designing smaller low-depth circuits for recursively-defined families of matrices like $H_n$. We show that a nontrivial synchronous circuit construction for any fixed matrix in the family leads to a smaller circuit for every matrix in the family. 

\begin{lemma} \label{intro:part1}
Let $M \in \F^{q \times q}$ be a $q \times q$ matrix over any field $\F$, and suppose there are matrices $A_1, \ldots, A_{d}$ such that $M = \prod_{j=1}^{d} A_j$ and $\nnz(A_i) \leq q^c$ for all $i \in [d]$. Then, for every positive integer $n$, letting $N = q^n$, the $N \times N$ matrix $M^{\otimes n}$ has a depth-$d$ synchronous linear circuit of size $O(N^c)$.
\end{lemma}
\noindent \Cref{intro:part1} follows by simply calculating how taking a Kronecker power changes the given circuit for $M$, but it is nonetheless conceptually interesting: in order to design a small circuit for the entire family of matrices $M^{\otimes n}$, it suffices to design one for any fixed matrix in the family. \Cref{intro:part1} is similar to the approach for designing matrix multiplication algorithms spearheaded by Strassen~\cite{strassen1969gaussian}, where an identity for quickly multiplying fixed size matrices implies asymptotic improvements for multiplying matrices of any sizes. Our proof was inspired by this, as Kronecker products also play a central role in the definition and study of matrix multiplication tensors.

We then use rigidity upper-bounds for the $q \times q$ matrix $M$ to construct fixed upper bounds. One can see by concatenating the two parts of a non-rigidity expression for $M$ that, for any rank $r$, we can find matrices $B,C$ with $M = B \times C$, $\nnz(B) = q(r+1)$, and $\nnz(C) = q\cdot r + \rig_M(r)$. We can `symmetrize' this construction using a Kronecker product trick, then apply \Cref{intro:part1} to yield:

\begin{lemma} \label{intro:rigiditytoconstruction}
Let $M \in \F^{q \times q}$ be a $q \times q$ matrix over any field $\F$, and $1 \leq r \leq q$ be any rank, and define $$c := \log_q( (r+1) \cdot (r + \rig_M(r)/q)).$$ Then, for any positive integer $n$, setting $N = q^n$, the $N \times N$ matrix $M^{\otimes n}$ has a depth-$d$ synchronous circuit of size $O(d \cdot N^{1 + c/d})$.
\end{lemma}

Thus, rigidity upper bounds on $H_m$ for a fixed $m$ can give nontrivial low-depth circuit upper bounds for $H_n$ for all $n$. Unfortunately, we cannot simply substitute in the rigidity upper bound of~\cite{alman2017probabilistic} to prove our result. Indeed, to achieve $c < 1$ in \Cref{intro:rigiditytoconstruction} when applying it to the $q \times q$ matrix $H_m$ for $q = 2^m$, it is not hard to see that we need $r < \sqrt{q}$. By comparison, the bound from~\cite{alman2017probabilistic} is primarily interesting for higher rank $r > q^{1 - \eps'}$ for small $\eps'>0$. Other known constructions, including those from probabilistic polynomials~\cite{alman2015probabilistic}, do not seem to give a nontrivial bound here either. Instead, to prove our upper bound, we use a new rigidity upper bound for $H_n$ for rank $r=1$, and more specifically, \Cref{intro:mainthm} ultimately follows from a new construction we give for the $16 \times 16$ matrix $H_4$ showing that $\rig_{H_4}(1) \leq 96$.

Using rigidity upper bounds naturally leads to `symmetric' circuits to use in \Cref{intro:part1}, but one could imagine other approaches that lead to more `lopsided' constructions. We additionally prove a generalization of \Cref{intro:part1}, that even such constructions can lead to upper bounds for $M^{\otimes n}$ for all $n$:
\begin{lemma} \label{intro:strassenylemma}
Let $M \in \F^{q \times q}$ be a $q \times q$ matrix over any field $\F$, and suppose there are matrices $A_1, \ldots, A_{d}$ such that $M = \prod_{j=1}^{d} A_j$, which is nontrivial in the sense that $\prod_{j=1}^{d} \nnz(A_j) \leq q^{d + c'}$ for some $c'<1$. Then, for every positive integer $n$, letting $N = q^n$, the $N \times N$ matrix $M^{\otimes n}$ has a depth-$d$ synchronous circuit of size $O(N^{1 + c/d})$ for a constant $c<1$ which depends only on $c'$.
\end{lemma}
Note that one could achieve $c'=1$ in \Cref{intro:strassenylemma} trivially by picking $A_1 = M$ and $A_2 = \cdots = A_d = I_q$, the $q \times q$ identity matrix. \Cref{intro:strassenylemma} shows that any construction which improves on this at all leads to an asymptotically smaller circuit for $M^{\otimes n}$. While \Cref{intro:part1} required that each $A_i$ has $\nnz(A_i) < q^{1 + 1/d}$, \Cref{intro:strassenylemma} instead only requires that the geometric mean of all the $\nnz(A_i)$ is less than $q^{1 + 1/d}$. However, it results in a slightly worse final size bound, which is why we use \Cref{intro:part1} to prove \Cref{intro:mainthm}.

\subsection{Surpassing Other Bounded-Coefficient Lower Bounds?}

It is natural to ask next whether our techniques can be used to overcome other bounded-coefficient lower bounds. We discuss a few more:

    \paragraph{Unbounded-Depth Circuits for $H_n$} Pudl{\'a}k~\cite{pudlak2000note} also showed a lower bound against unbounded-depth bounded-coefficient synchronous linear circuits for computing $H_n$. 
    \begin{theorem}[\cite{pudlak2000note}]  \label{pudlak2}
    Any synchronous linear circuit with $c$-bounded coefficients for computing the Walsh-Hadamard transform $H_n \in \C^{N \times N}$ for $N=2^n$ has size $\geq \frac{e \cdot \log_e(2)}{c^2}N \log_2 N$. 
    \end{theorem}
    For $c=1$ (as is the case in all previous circuits for $H_n$), this gives a lower bound of $e \cdot \log_e(2) \cdot N \log_2 N \approx 1.884 \cdot N \log_2 N$. This is known to be tight, as optimizing for $d$ in the usual fast Walsh-Hadamard transform gives a matching upper bound. In fact, we give a new construction which also beats this lower bound, although only by a constant factor.
    \begin{theorem} \label{intro:unboundedsizeub}
    Let $\F$ be any field, and let $M \in \F^{2 \times 2}$ be any matrix over $\F$. There is a constant $\eps > 0.01526$ such that, for any positive integer $n$, the linear transformation $M^{\otimes n} \in \F^{N \times N}$ for $N = 2^n$ has a synchronous linear circuit of size $(1 - \eps + o(1)) \cdot e \cdot \log_e(2) \cdot N \log_2 N$. When $M = H_1$, so that $M^{\otimes n}$ is the Walsh-Hadamard transform $H_n$, we can improve the bound to $\eps > 0.04816$. 
    \end{theorem}
    It is no coincidence that our bounds on $\eps$ in \Cref{intro:unboundedsizeub} are the same as those in \Cref{intro:mainthm}: We prove \Cref{intro:unboundedsizeub} by introducing a gadget which increases the depth in \Cref{intro:mainthm} but removes the additional unwanted $2^{\eps}$ term in the circuit size (which would otherwise impact our constant-factor savings), and then optimizing over all choices of $d$.
    
    Of course, it would be much more exciting to design a circuit of size $o(N \log N)$ for $H_n$, but that is currently beyond our techniques. That said, we believe \Cref{intro:unboundedsizeub} gives the first improvement of any kind on the standard fast Hadamard transform for computing $H_n$, and we are optimistic that further improvements are possible.

    \paragraph{Circuits for the Fourier Transform} Pudl{\'a}k showed that both \Cref{pudlak1} and \Cref{pudlak2} also hold for the Discrete Fourier transform\footnote{Morgenstern~\cite{morgenstern1973note} first showed such a result for linear circuits which need not be synchronous, with slightly lower leading constant factors.}, $F_N \in \C^{N \times N}$. Can our approach be used to beat these lower bounds as well? We remark that $F_N$ is actually \emph{too rigid} for our approach using \Cref{intro:rigiditytoconstruction} to apply to overcome this bound. Interestingly, the rigidity lower bound we use to show this is not the asymptotically best known bound of $\rig_{F_N}(r) \geq \Omega(\frac{N^2}{r}\log(N/r))$, but instead the bound $\rig_{F_N}(r) \geq (N-r)^2/(r+1)$~\cite{Shparlinski} which has better known constant factors for small $r$.
    
    It should be noted that we do not rule out the existence of $o(d \cdot N^{1 + 1/d})$ size depth-$d$ linear circuits for $F_N$, or even rule out that \Cref{intro:part1} could be used to construct such circuits. However, an approach different from our non-rigidity approach would be needed to give the nontrivial construction needed by \Cref{intro:part1}.

    \paragraph{Matrix Multiplication} Raz~\cite{raz2002complexity} showed that any bilinear circuit with bounded coefficients for computing the product of two $N \times N$ matrices over $\C$ requires size $\Omega(N^2 \log N)$. This is not known to be tight: the best known circuit for $N \times N \times N$ matrix multiplication has size $N^{\omega + o(1)}$ where $\omega \leq 2.373$~\cite{williams2012multiplying,le2014powers,alman2020refined} is the matrix multiplication exponent. That said, as we will discuss soon in more detail in \Cref{introsec:mmult}, there is a strong connection between this lower bound and the aforementioned bounded-coefficient lower bounds: if one could surpass Raz's lower bound and design an $o(N^2 \log N)$ size circuit for matrix multiplication, it would lead to linear circuits of size $o(N \log N)$ for both the $N \times N$ discrete Fourier transform and the $N \times N$ Walsh-Hadamard transform, as well as many related linear transformations.

\subsection{More Matrices Are Not Valiant-Rigid}

Our next upper bound is a new non-rigidity result, which generalizes and sheds new light on the non-rigidity of the Walsh-Hadamard transform~\cite{alman2017probabilistic}. We focus on two families of matrices $M$ which generalize $H_n$.

\begin{enumerate}
    \item Matrices $M \in \F^{q^n \times q^n}$ of the form $M = \bigotimes_{\ell=1}^n M_i$ for positive integers $q,n$ and any matrices $M_1, \ldots, M_n \in \F^{q \times q}$ (where $\otimes$ denotes the Kronecker product). Kronecker power matrices like $H_n$ which we discussed earlier are of this form with $M_1 = M_2 = \cdots = M_n$, but here we also allow for different choices of the matrices $M_1, \ldots, M_n$. 
    
    \item Matrices $M \in \F^{q^n \times q^n}$ whose entries are given by, for $x,y \in \{0,1,\ldots,q-1\}^n$: $$M[x,y] = f(\max\{ x[1],y[1] \}, \max\{ x[2],y[2] \}, \max\{ x[3],y[3] \}, \ldots, \max\{ x[n],y[n] \})$$ for \emph{any} function $f : \{0,1,\ldots,q-1\}^n \to \F$. For instance, $H_n$ is of this form with $q=2$ when $f$ is the parity function, but we also allow for more complicated choices of $f$.
\end{enumerate}

\begin{theorem}\label{intro:mainnonrig}
Any matrix of either of the above forms with $q \leq O(\log n)$ is not Valiant-rigid. More precisely, setting $N = q^n$, any such $M$ satisfies, for any sufficiently small $\eps>0$: $$\rig_M(N^{1 - \frac{q}{2^q} \cdot O(\eps^2 / \log^2(1/\eps))}) \leq N^{1+\eps}.$$
\end{theorem}

The constant hidden by the $O$ in \Cref{intro:mainnonrig} is not too small; for instance, we show that when $q=2$, any such $M$ has $\rig_M(O(N^{0.981})) < o(N^2)$.

\Cref{intro:mainnonrig} shows that it was not just a `coincidence' that $H_n$ is not rigid, but in fact a number of big families of matrices generalizing $H_n$ are also not rigid. It, of course, rules out the Valiant-rigidity approach for proving circuit lower bounds for any of these linear transformations. We now discuss the two families of matrices in some more detail.

\begin{enumerate}
    \item Aside from being a natural generalization of $H_n$, Kronecker products like this are ubiquitous in many areas of computational science (see e.g.~\cite{van2000ubiquitous}). The non-rigidity of these matrices is also interesting compared with our observation which we discuss in detail in the upcoming \Cref{introsec:mmult} that: if there are Valiant-rigid matrices in this family for any fixed $n$ and growing $q$, then we would get a lower bound for $N \times N \times N^n$ matrix multiplication. By comparison, \Cref{intro:mainnonrig} shows there are no Valiant-rigid matrices in this family for fixed $q$ and growing $n$. The difference between this family of matrices when $n$ is growing versus when $q$ is growing is not unlike the difference between the families of Walsh-Hadamard transforms and Fourier transforms (which are both Hadamard matrices for different choices of which of the two defining parameters is growing). Perhaps the techniques of~\cite{dvir2019fourier} for showing that Fourier transforms are not rigid could help to approach this other setting.
    
    \item As noticed by~\cite{alman2017probabilistic}, matrices of this form for different choices of the function $f : \{0,1,\ldots,q-1\}^n \to \F$ arise frequently in fine-grained complexity, especially in the case $q=2$. In fact, the best known algorithms for a number of different problems have used, as their key insight, that this type of matrix $M$ is not rigid, including the Orthogonal Vectors problem~\cite{abboud2014more} (for $f = AND$), All-Pairs Shortest Paths~\cite{williams2014faster} (also for $f = AND$), and Hamming Nearest Neighbors~\cite{alman2015probabilistic, alman2016polynomial} (for $f = MAJORITY$). These algorithms all use the `polynomial method' to show that $M$ is not rigid in a low-rank, high-error regime, but it is unclear how to extend them to less structured functions $f$. By comparison, \Cref{intro:mainnonrig} shows that $M$ is not rigid in a higher-rank, lower-error regime, and it applies to \emph{any} function $f$.
    
    In fact, in addition to these aforementioned algorithms, all the prior work on showing that matrices of interest are not Valiant-rigid~\cite{alman2017probabilistic,dvir2017matrix,dvir2019fourier} has used the polynomial method. For instance, the previous proof of the non-rigidity of the Walsh-Hadamard transform~\cite{alman2017probabilistic} critically used the fact that the corresponding function $f = PARITY$ has low-degree polynomial approximations (which are correct on most inputs) over any field. Our rigidity upper bound \emph{does not} use the polynomial method (at least explicitly), and applies to \emph{any} function $f$ without any restriction on how well it can be approximated by polynomials. In other words, this central property of $f$ that was used by prior work is actually unnecessary for proving that $M$ is not Valiant-rigid.
\end{enumerate}

Our proof of \Cref{intro:mainnonrig} in the case $q=2$ is actually quite simple, and it simplifies the previous proof of the non-rigidity of the Walsh-Hadamard transform. Inspired by Dvir and Liu~\cite{dvir2019fourier}, who frequently make use of the fact that the product of a constant number of matrices which are not Valiant-rigid is, itself, not Valiant-rigid (see \Cref{lem:productnotrigid} below), we begin by noticing that any matrix $M$ from either of the two families can be written as \begin{align}\label{eqn:decomp}M = D \times R_n \times D' \times R_n \times D'',\end{align}
where $D, D', D'' \in \F^{2^n \times 2^n}$ are three carefully-chosen diagonal matrices (which are evidently not Valiant-rigid), and $R_n \in \{0,1\}^{2^n \times 2^n}$ is the \emph{disjointness} matrix, given by $R_n := R_1^{\otimes n}$ where $$R_1 := 
\begin{bmatrix}
1 & 1 \\
1 & 0
\end{bmatrix}.$$
Thus, to show that \emph{any} such $M$ is not Valiant-rigid, it suffices to show that $R_n$ is not Valiant-rigid. However, this is not too difficult, since $R_n$ is a fairly sparse matrix to begin with! Indeed, $R_n$ is a $2^n \times 2^n$ matrix, but has only $3^n$ nonzero entries. Moreover, most of these nonzero entries are concentrated in a few rows and columns: for each integer $0 \leq k \leq n$, the matrix $R_n$ has $\binom{n}{k}$ rows (or columns) with $2^k$ nonzero entries. Using standard bounds on binomial coefficients, we thus see that, by removing only the $2^{n(1 - \Theta(\eps^2 / \log^2 (1/\eps)))}$ densest rows and columns of $R_n$, we are left with a matrix with only $2^{n \cdot \eps}$ nonzero entries per row or column. Since changing a single row or column of a matrix is a rank-1 update, this shows that $R_n$ is not Valiant-rigid as desired. 

Extending this result to larger $q$ is quite a bit more involved. Let us focus for now on family 1 of matrices above (Kronecker products of $n$ different $q \times q$ matrices); the proof for family 2 is similar. We will proceed by induction on $q$. Our starting point is the remark that any $q \times q$ matrix $M_i$ can be written as the sum of a $q \times q$ rank-1 matrix $J_i$, and a $(q-1) \times (q-1)$ matrix $L_i$ (padded with a row and column of $0$s). For instance, in the case $q=3$ we have (assuming the top-left entry $a$ is nonzero):
$$
\begin{bmatrix}
a & b & c\\
d & e & f\\
g & h & i\\
\end{bmatrix}
=
\begin{bmatrix}
a & b & c\\
d & \frac{bd}{a} & \frac{bc}{a}\\
g & \frac{bg}{a} & \frac{bc}{a}\\
\end{bmatrix}
+
\begin{bmatrix}
0 & 0 & 0\\
0 & e-\frac{bd}{a} & f-\frac{bc}{a}\\
0 & h-\frac{bg}{a} & i-\frac{bc}{a}\\
\end{bmatrix}.$$
We have now written $M_i = J_i + L_i$, and we know that $\bigotimes_{i=1}^n J_i$ is not Valiant-rigid (in fact, it has rank $1$), and $\bigotimes_{i=1}^n L_i$ is not Valiant-rigid, even when thought of as a $(q-1)^n \times (q-1)^n$ matrix, by the inductive hypothesis. This does not imply that $\bigotimes_{i=1}^n M_i$ is not Valiant-rigid on its own, however, because there are cross-terms:
$$\bigotimes_{i=1}^n M_i = \bigotimes_{i=1}^n (J_i + L_i) = \sum_{K \subseteq \{1,2,\ldots,n\}} \bigotimes_{i=1}^n \left( [i \in K] ~?~ L_i : J_i \right)$$
(Here, we are using $\left( [i \in K] ~?~ L_i : J_i \right)$ as the ternary operator, which equals $L_i$ when $i \in K$, and equals $J_i$ when $i \notin K$).  For any particular $K$, the matrix $M_K := \bigotimes_{i=1}^n \left( [i \in K] ~?~ L_i : J_i \right)$ can be seen as the Kronecker product of a $q^{|K|} \times q^{|K|}$ matrix of rank $1$, and a $q^{n-|K|} \times q^{n-|K|}$ matrix which, by the inductive hypothesis, is not Valiant-rigid. It can be shown (see e.g.~\cite[{Section~6}]{dvir2019fourier}) that the Kronecker product of matrices which are not Valiant-rigid is itself not Valiant-rigid, and hence that $M_K$ is not Valiant-rigid. However, this is still not sufficient: we have now only expressed $M$ as the sum of $2^n$ matrices which are not Valiant-rigid, but whose sum might still be.

We instead first perform a number of low-rank updates to $M$ to simplify the problem. We first subtract away all the matrices $M_K$ for which $|K|$ is not close to $(q-1)n/q$. Next, we remove all rows and columns corresponding to $x \in \{0,1,\ldots,q-1\}^n$ for which $\nnz(x)$ is not close to $(q-1)n/q$. Finally, we observe that each remaining row of $M$ only intersects with a nonzero row of $q^{O(\eps \cdot n)}$ different choices of remaining matrices $M_K$ (compared with $q^n$ before). Hence, the fact that each $M_K$ is not Valiant-rigid implies our desired non-rigidity, as the sparsity per row is now only multiplied by $q^{O(\eps \cdot n)}$. We have, of course, glossed over many important and intricate aspects of the proof; we refer the reader to \Cref{sec:bigproof} for the details.

We briefly remark that the techniques for manipulating Kronecker products used by Dvir and Liu~\cite{dvir2019fourier} do not appear sufficient to prove our \Cref{intro:mainnonrig}. They observed that the Kronecker product of matrices $M_1, \ldots, M_n$ which are not Valiant-rigid is itself not Valiant-rigid. In particular, they begin with a decomposition $M_i = J_i + L_i$ where $J_i$ has low rank like in our setting, but they further assume that $L_i$ is very sparse. In our case, $M_1, \ldots, M_n$ are arbitrary matrices, and may all be very rigid on their own, and so a more intricate argument seems necessary.

\subsection{Connections Between Matrix Multiplication and Kronecker Product Linear Transformations} \label{introsec:mmult}
    
    As we previously mentioned, Raz~\cite{raz2002complexity} showed that any bilinear circuit with bounded coefficients for computing the product of two $N \times N$ matrices over $\C$ requires size $\Omega(N^2 \log N)$. A key insight behind Raz's lower bound is that, for a fixed matrix $A \in \F^{N \times N}$, the following two problems are equivalent:
    \begin{itemize}
        \item Given as input a matrix $B \in \F^{N \times N}$, output the matrix $A \times B$.
        \item Given as input a vector $b \in \F^{N^2}$, output the linear transformation $(I_N \otimes A) b$.
    \end{itemize}
    In particular, if one could show that there is any matrix $A \in \F^{N \times N}$ for which the linear transformation $I_N \otimes A \in \F^{N^2 \times N^2}$ does not have $O(N^2)$ size circuits, then $N \times N \times N$ matrix multiplication does not have $O(N^2)$ size circuits. One intriguing avenue toward showing this is to show that there exists an $A \in \F^{N \times N}$ such that $I_N \otimes A$ is \emph{Valiant-rigid}. In contrast with the usual setting in matrix rigidity, here, to show a lower bound against a particular problem (matrix multiplication), it suffices to show that there exists a rigid matrix among a large family of matrices. (Roughly, Raz's lower bound is proved by showing there exists an $A \in \F^{N \times N}$ such that $I_N \otimes A$ has a high value of a variant of rigidity which corresponds to bounded-coefficient circuits.)
   
   We take this observation further, showing that there is a much larger family of matrices for which a circuit lower bound would imply lower bounds for matrix multiplication. The key idea is the following algorithm for using matrix multiplication to compute linear transformations defined by Kronecker products (which is not very difficult to prove, and is likely folklore):
    \begin{proposition} \label{intro:mmultprop}
    For any field $\F$, and any fixed positive integer $k$, suppose that $N \times N \times N^{k-1}$ matrix multiplication over $\F$ has an arithmetic circuit of size $o(N^{k} \log N)$. Then, the $N \times N$ Fourier transform, $N \times N$ Walsh-Hadamard transform, and any transform which can be written as the Kronekcer product of $k$ different $N^{1/k} \times N^{1/k}$ size matrices, have arithmetic circuits of size $o(N \log N)$. 
    \end{proposition}
    Applying \Cref{intro:mmultprop} with $k=2$, we see that if one shows there are any matrices $A,B \in \F^{N \times N}$ such that $A \otimes B \in \F^{N^2 \times N^2}$ requires circuits of size $\Omega(N^2 \log N)$ (perhaps making use of a proof that $A \otimes B$ is Valiant-rigid\footnote{Actually, showing that $A \otimes B$ is Valiant-rigid would only prove a $\omega(N^2)$ lower bound against $O(\log N)$-depth circuits for $N \times N \times N$ matrix multiplication. Normally, a $O(\log N)$ depth restriction on circuits for $N \times N \times N$ matrix multiplication is not very limiting, since it is known that arithmetic circuits for matrix multiplication can be converted into logarithmic-depth circuits with only a $O(N^\eps)$ blowup in size for any $\eps>0$ (which, in particular, does not effect the value of the matrix multiplication exponent $\omega$). However, in our setting where the resulting lower bounds are only for size $\Omega(N^2 \log N)$, this $N^\eps$ term may be non-negligible.}, or in some other way), then $N \times N$ matrix multiplication requires circuits of size $\Omega(N^2 \log N)$. By comparison, even for very simple matrices of the form $A \otimes B$ such as the $N^2 \times N^2$ Discrete Fourier transform or Walsh-Hadamard transform, the best known circuit size is only $\Theta(N^2 \log N)$.
    
    \Cref{intro:mmultprop} becomes more exciting from an algorithmic perspective as we consider larger $k$. For $k=2$, the upper bound of $o(N^2 \log N)$ needed for $N \times N \times N$ matrix multiplication is quite far away from our current best upper bound of roughly $O(N^{2.373})$. However, as $k$ grows, the exponent is known to approach $k$ as well:

\begin{proposition}[\cite{huang1998fast}] \label{intro:rectmm}
For every field $\F$ and integer $k>1$, there is a circuit of size $O(N^{k \cdot \log_{k-1}(k)})$ for performing $N \times N \times N^{k-1}$ matrix multiplication. Here, the $O$ is hiding a function of $k$. Note that the exponent is $$k \cdot \log_{k-1}(k) = k + O\left(\frac{1}{\log k}\right).$$
\end{proposition}

In fact, working through the details (see \Cref{sec:mmult} below), we find that for a slightly super-constant choice of $k = \log N / \log\log N$, a circuit of size $O(N^{k \cdot \log_{k-1}(k)})$ for $N \times N \times N^{k-1}$ matrix multiplication would lead to an $o(N \log N)$ time algorithm for the $N \times N$ Fourier transform and the $N \times N$ Walsh-Hadamard transform.
Unfortunately, this is not exactly what is guaranteed to us by \Cref{intro:rectmm}; we only know there is such a circuit of size $f(k) \cdot N^{k \cdot \log_{k-1}(k)}$ for some function $f$. When $k$ is super-constant, the term $f(k)$, which is usually part of the leading constant in fast matrix multiplication algorithms, becomes relevant and may swamp our other savings. We show in \Cref{sec:mmult} below that any bound $f(k) < o(\log k)$ would suffice to speed up the $N \times N$ Fourier transform and the $N \times N$ Walsh-Hadamard transform. The growth of $f(k)$ in fast rectangular matrix multiplication algorithms is typically not the focus of study, as one typically thinks of $k$ as a constant\footnote{The only work proving something like a bound on $f(k)$ that the author is aware of is Williams'~\cite{williams2014faster} analysis of Coppersmith's~\cite{coppersmith1982rapid} rectangular matrix multiplication algorithm. He shows the algorithm for $N \times N^{0.17} \times N$ matrix multiplication has a running time of only $N^2 \polylog(N)$, compared to the bound of $O(N^{2+\eps})$ for any $\eps>0$ that one achieves using Coppersmith's identity combined with standard fast matrix multiplication techniques.}, but it may warrant further investigation!

\subsection{Fast Batch Computations on Low-Dimensional Points} \label{introsec:batches}

For our last new upper bound, we remark that some ideas in the proof of \Cref{intro:mainnonrig} can be used to extend certain algorithms for the Orthogonal Vectors problem (which corresponds to the disjointness matrix $R_n$) to a more general class of problems. Recall that in the Orthogonal Vectors problem, we are given as input $m$ vectors from $\{0,1\}^d$, and the goal is to determine whether there is a pair which is orthogonal (over $\Z$). Equivalently, we are given as input $m$ row and column indices into the matrix $R_d$, and we want to determine whether there are any $1$s in the corresponding submatrix. This can be solved in $O(m^2 \cdot d)$ time (and even faster when $d \leq O(\log m)$~\cite{abboud2014more}), but in the regime when $m \geq \tilde{\Omega}(2^{d/2})$, there is a faster folklore algorithm running in time only $O(m + d \cdot 2^d)$. In fact, this latter algorithm corresponds directly to the fact that the linear transformation $R_d$ can be computed in time $O(d \cdot 2^d)$. 

Using \Cref{eqn:decomp}, we can extend this to a more general class of problems, defined as follows. Let $f : \{0,1\}^d \to \F$ be a function which can be evaluated in time $T$. Then, given as input a set $S \subseteq \{0,1\}^d$ of size $|S|=m$, there is an algorithm running in time $O(m + (d+T) \cdot 2^d)$ for computing, for all $s \in S$, the sum $\sum_{t \in S} f(s[1] \wedge t[1], s[2] \wedge t[2], \ldots, s[d] \wedge t[d])$. When $f=NOR$, this algorithm counts the number of Orthogonal Vectors. However, other functions $f$ correspond to other interesting tasks. For instance, when $f$ is a threshold function (such as $MAJORITY$), this algorithm counts the number of pairs of points which share a certain number of $1$s in common, which is a basic nearest neighbor search problem, in time $O(m + d \cdot 2^d)$. This improves on the more straightforward $O(m \cdot 2^d)$ time algorithm for this problem when $d = o(m)$. 

\subsection{Other Related Work}

\paragraph{Rigidity Upper Bounds from Low-Depth Circuit Upper Bounds}
Our results discussed in \Cref{introsec:lowdepth} above show how rigidity upper bounds for a matrix $M$ can be used to construct small low-depth circuits for $M$. Relatedly, Pudl{\'a}k~\cite{pudlak1994communication} showed a type of converse: that low-depth circuit upper bounds can be used to show rigidity upper bounds.

\begin{proposition}[{\cite[Proposition~2]{pudlak1994communication}}]
For any field $\F$, positive integers $r,d$, real $c,\eps \geq 0$ and $M \in \F^{N \times N}$, if $M$ has a depth-$d$ linear circuit of size $O(d \cdot N^{1 + c/d})$, then $\rig_M(\eps \cdot N) \leq (d/\eps)^d \cdot N^{1+c}$.
\end{proposition}

Although this can be combined with our \Cref{intro:mainthm} to prove rigidity upper bounds for $H_n$ and other Kronecker power matrices, the resulting bounds are weaker than what we prove in \Cref{intro:mainnonrig} using a different approach, and do not suffice to prove that these matrices are not Valiant-rigid. Perhaps there is a different way to reconcile the two?

\paragraph{Data Structures and Rigidity}
Rigidity upper bounds are known to give rise to data structure bounds: Dvir, Golovnev, and Weinstein~\cite{dvir2019static} recently showed this for static data structures, and Ramamoorthy and Rashtchian~\cite{ramamoorthy2019equivalence} showed this for systematic linear data structures.

\paragraph{Small Depth Circuit Lower Bounds}

The best-known lower bounds on the size of a depth-$2$ linear circuit for computing an explicit $N \times N$ linear transformation are only $\Omega(N \log^2 N / (\log \log N)^2)$ for efficient error-correcting codes over constant-size finite fields~\cite{gal2012tight}, or $\Omega(N \log^2 N / \log \log N)$ for matrices arising from super-concentrator graphs over larger fields~\cite{radhakrishnan2000bounds}. Two recent lower bounds were also shown for less-explicit matrices: Kumar and Volk~\cite{kumar2019lower} constructed a matrix in time $\exp(N^{\Theta(1)})$, over a field of size $\exp(N^{\Theta(1)})$, which requires depth-$d$ circuits of size $N^{1 + 1/(2d)}$. With Chen~\cite{alman2019efficient}, we construct a matrix in $P^{NP}$ which has $\{0,1\}$ entries over any fixed-size finite field and which requires depth-$2$ circuits of size $\Omega(N \cdot 2^{(\log N)^{1/4 - \delta}})$ for any $\delta>0$. In other words, the known techniques are far from proving that any of the depth-$d$ upper bounds presented here, which are of the form $O(N^{1 + (1-\eps)/d})$ for somewhat small constants $\eps>0$, are tight.

\paragraph{Other Circuit Models for Matrices}

Circuit models other than linear circuits have also been studied for computing matrices in certain settings. For instance, when working with matrices over a semigroup (like the OR semigroup) or a semiring (like the SUM semiring) instead of a field, one can consider circuits where the gates compute sums from that semigroup or semiring instead. See, for instance, the book by Jukna and Sergeev which studies these models in detail~\cite{jukna2013complexity}. These models have applications to areas like communication complexity, and the techniques for constructing circuits in these models often apply to the linear circuit model as well. For instance, we remark in \Cref{sec:disjconst} below that a construction by Jukna and Sergeev for the disjointness matrix $R_n$, which takes advantage of both the recursive definition and the sparsity of $R_n$, leads to a better upper bound for low-depth circuits for $R_n$ than we are able to prove using our rigidity approach.

\subsection{Outline}

In Section 2, we introduce the notions and notation we will use, and we present a number of basic tools for working with Kronecker products and linear circuits. We then prove \Cref{intro:mainthm} in Sections 3 and 4: we prove \Cref{intro:part1} and \Cref{intro:rigiditytoconstruction} in Section 3, and then we study low-rank rigidity upper bounds for a number of families of matrices in Section 4. In Sections 5--7 we prove \Cref{intro:mainnonrig}: we prove that $R_n$ is not Valiant-rigid in Section 5, we show how to express other matrices of interest in terms of $R_n$ in Section 6, and we give our extension to Kronecker products of larger matrices (the $q>2$ case of \Cref{intro:mainnonrig}) in Section 7. Finally, in Section 8 we investigate connections between the linear complexity of Kronecker products and matrix multiplication, in Section 9 we present other algorithms which we design using ideas from the remainder of the paper, and in Section 10 we prove \Cref{intro:strassenylemma}, the generalization of \Cref{intro:part1}.

\section{Preliminaries}

\subsection{Notation and Basic Properties}

\subsubsection{Matrix Indexing}

For a positive integer $n$, we write $[n] := \{1,2,\ldots,n\}$ and $[n]_0 := \{0,1,\ldots,n-1\}$.

By default, we use zero-based numbering for the indices of matrices, meaning, for any set $S$, positive integers $n,m$, matrix $M \in S^{n \times m}$, $i \in [n]_0$ and $j \in [m]_0$, we write $M[i,j]$ for the corresponding entry of $M$. That said, if $S_n, S_m$ are sets of sizes $|S_n|=n$ and $|S_m|=m$, we may sometimes say that the rows and columns of $M$ are \emph{indexed by} $S_n$ and $S_m$, respectively. In this case, we implicitly define bijections $f_{S_n} : S_n \to [n]_0$ and $f_{S_m} : S_m \to [m]_0$, and then for $s_n \in S_n$ and $s_m \in S_m$ we write $M[s_n, s_m] := M[f_{S_n}(s_n), f_{S_m}(s_m)]$.

\subsubsection{Matrix Products}

\begin{definition}
For any field $\F$, positive integers $n_A,n_B,m_A,m_B$, and matrices $A \in \F^{a_1 \times a_2}$, $B \in \F^{b_1 \times b_2}$, the \emph{Kronecker product} of $A$ and $B$, denoted $A \otimes B$, is the matrix $A \otimes B \in \F^{(a_1 \cdot b_1) \times (a_2 \cdot b_2)}$, whose rows and columns are indexed by $[a_1]_0 \times [b_1]_0$ and $[a_2]_0 \times [b_2]_0$, respectively, and whose entries are given by $$A \otimes B[(i_A,i_B),(j_A,j_B)] := A[i_A,j_A] \cdot B[i_B,j_B].$$

The Kronecker product is not commutative in general, however, there are always permutation matrices $P \in \{0,1\}^{(a_1 \cdot b_1) \times (a_1 \cdot b_1)}$ and $P' \in \{0,1\}^{(a_2 \cdot b_2) \times (a_2 \cdot b_2)}$, which depend only on $a_1, a_2, b_2$, and $b_2$, such that $$A \otimes B = P \times (B \otimes A) \times P'.$$

For a matrix $A$ and positive integer $n$, we write $A^{\otimes n}$ to denote the Kronecker product of $n$ copies of $A$, i.e., $A^{\otimes 1} = A$ and $A^{\otimes n} = A^{\otimes (n-1)} \otimes A$.

We will need some additional notation for dealing with more complicated Kronecker products. For positive integers $n,q$, matrices $A,B \in \F^{q \times q}$, and sets $S_A \subseteq [n]$ and $S_B = [n] \setminus S_A$, we write $A^{\otimes S_A} \otimes B^{\otimes S_B}$ for the matrix in $\F^{q^n \times q^n}$ given by, for $i,j \in [q]_0^n$,
$$A^{\otimes S_A} \otimes B^{\otimes S_B}[i,j] := \left( \prod_{\ell \in S_A} A[i[\ell],j[\ell]] \right) \cdot \left( \prod_{\ell \in S_B} B[i[\ell],j[\ell]] \right).$$

Similarly, if $A \in \F^{q \times q}$ and $B \in \F^{q^{|S_B|} \times q^{|S_B|}}$ then we write $A^{\otimes S_A} \otimes B^{\otimes S_B}$ for the matrix in $\F^{q^n \times q^n}$ given by, for $i,j \in [q]_0^n$,
$$A^{\otimes S_A} \otimes B^{\otimes S_B}[i,j] := \left( \prod_{\ell \in S_A} A[i[\ell],j[\ell]] \right) \cdot \left( B[i|_{S_B}, j|_{S_B}] \right).$$ Here, `$i|_{S_B}$' denotes $i$ restricted to the coordinates of $S_B$.
\end{definition}

In addition to using $\otimes$ to denote the Kronecker product of matrices, we will use $\times$ to denote the (usual) product of matrices, and for emphasis, we will use $\cdot$ to denote the product of field elements.

\subsubsection{Matrix Sparsity and Rigidity}

For a matrix $A \in \F^{a_1 \times a_2}$, its sparsity, written $\nnz(A)$, denotes number of non-zero entries in $A$. We similarly define its row sparsity, $\nnzr(A)$, to be the maximum number of non-zero entries in a row of $A$, and its column sparsity, $\nnzc(A)$, to be the maximum number of non-zero entries in a column of $A$. Some basic properties we will use are that, for any $A \in \F^{a_1 \times a_2}$ and $B \in \F^{b_1 \times b_2}$:
\begin{itemize}
    \item $\nnz(A \otimes B) = \nnz(A) \cdot \nnz(B)$,
    \item $\nnzr(A \otimes B) = \nnzr(A) \cdot \nnzr(B)$,
    \item if $a_2=b_1$ then $\nnzr(A \times B) \leq \nnzr(A) \cdot \nnzr(B)$,
    \item if $a_2=b_1$ then $\nnz(A \times B) \leq \nnz(A) \cdot \nnzr(B)$, and
    \item if $D \in \F^{a_1 \times a_1}$ is a diagonal matrix, then $\nnz(D \times A) \leq \nnz(A)$ and $\nnzr(D \times A) \leq \nnzr(A)$.
\end{itemize}

For a matrix $A \in \F^{a \times a}$ and a nonnegative integer $r$, we write $\rig_A(r)$ to denote the rank-$r$ rigidity of $A$ over $\F$, which is the minimum number of entries of $A$ which must be changed to other values in $\F$ to make its rank at most $r$. In other words:
$$\rig_A(r) := \min_{\overset{B \in \F^{a \times a},}{\rank(A+B) \leq r}} \nnz(B).$$
The definition of $\rig_A(r)$ depends on the field $\F$, which we will explicitly mention when it is not clear from context.

We similarly define the rank-$r$ row/column rigidity of $A$, denoted $\rig_A^{rc}(r)$, to be the minimum number of entries which must be changed per row or column of $A$ to make its rank at most $r$, i.e.
$$\rig_A^{rc}(r) := \min_{\overset{B \in \F^{a \times a},}{\rank(A+B) \leq r}} \max\{ \nnzr(B), \nnzc(B) \}.$$

It follows that, for any positive integer $r$, and any $A \in \F^{a \times a}$, we have $$\rig_A(r) \leq a \cdot \rig_A^{rc}(r).$$
\subsubsection{Important Families of Matrices} \label{subsubsec:matricesdefined}

\begin{itemize}
\item The family of Walsh-Hadamard transforms, $H_n \in \{-1,1\}^{2^n \times 2^n}$, is defined by $$H_1 = \begin{bmatrix}
1 & 1 \\
1 & -1
\end{bmatrix}$$
and for $n \in \N$, $H_n = H_1^{\otimes n}$. 

\item The family of Disjointness matrices, $R_n \in \{0,1\}^{2^n \times 2^n}$, is defined by $$R_1 = \begin{bmatrix}
1 & 1 \\
1 & 0
\end{bmatrix}$$
and for $n \in \N$, $R_n = R_1^{\otimes n}$. 

\item The family of Fourier transforms, $F_N \in \C^{N \times N}$, is defined by picking $\omega_N := e^{2 \pi i/N}$ to be a primitive $N$th root of unity, then setting $F_N[i,j] = \omega_N^{i \cdot j}$.

\item For $k \in \N$ we write $I_k$ to denote the $k\times k$ identity matrix. 

\item A diagonal matrix $D \in \F^{N \times N}$ is any matrix such that, if $i \neq j$, then $D[i,j]=0$. $D$ has full rank if and only if $D[i,i]\neq 0$ for all $i$.

\item A weighted permutation matrix $\Pi \in \F^{N \times N}$ is a matrix with exactly one nonzero entry in each row or column. A permutation matrix is a weighted permutation matrix in which each nonzero entry is $1$.
\end{itemize}

\subsubsection{Arithmetic Circuits and Linear Circuits}

An \emph{arithmetic circuit} over a field $\F$ is a circuit whose inputs are variables and constants from $\F$, and whose gates compute the product or the sum over $\F$ of their inputs. A \emph{linear circuit} over $\F$ is a circuit whose inputs are variables from $\F$, and whose gates compute $\F$-linear combinations of their inputs. The depth of a circuit is the length (number of edges) of the longest path from an input to an output. The size might either be measured by number of gates, or number of wires. 

For a field $\F$ and matrix $A \in \F^{q_1 \times q_2}$, we say that a circuit $C$ computes the linear transformation $A$ (or simply `computes $A$') if $C$ has $q_1$ inputs and $q_2$ outputs, such that on input $x \in \F^{q_1}$, the output of $C$ is $A \times x$.

In a \emph{synchronous} linear circuit, the inputs to each gate must all have the same depth. A synchronous linear circuit $C$ of depth $d$ for a matrix $A$ corresponds to matrices $A_1, \ldots, A_d$ such that $A = \prod_{j=1}^d A_j$, and the size (number of wires) of $C$ is given by $\sum_{j=1}^d \nnz(A_j)$. Any depth-$d$ linear circuit can be converted into a depth-$d$ synchronous linear circuit for the same linear transformation with at most a $O(d)$ multiplicative blow-up in the size. In this paper, $O(d)$ will typically be negligible, so we will focus on synchronous linear circuits.

\subsubsection{Binary Entropy Function}

The binary entropy function $H : [0,1] \to [0,1]$ is defined by
$$H(p) := -p \cdot \log_2(p) - (1-p) \cdot \log_2(1-p),$$
where we take $0 \cdot \log_2(0) = 0$.
For every integer $n>1$ and every $p \in (0,1)$, it is known that
$$\frac{1}{n+1} 2^{n \cdot H(p)} \leq \binom{n}{p \cdot n} \leq 2^{n \cdot H(p)}.$$

We will make use of the following calculations:

\begin{lemma}\label{binentropy}
For any integer $q>1$ and any real $0 < \delta< 1/q - 1/(q+1)$ we have:
\begin{enumerate}
    \item $H(1/q) = \log_2(q) - \frac{q-1}{q}\log_2(q-1)$,
    \item $H(1/q + \delta) - H(1/q) \leq \delta \cdot \log_2(q-1) - \delta^2 \cdot \frac{q^2}{(q-1)\log_e(4)} + O(\delta^3)$, and
    \item $H(1/q) - H(1/q - \delta) \leq \delta \cdot \log_2(q-1) + \delta^2 \cdot \frac{q^2}{(q-1)\log_e(4)} + O(\delta^3)$.
\end{enumerate}
\end{lemma}

\begin{proof}
(1) is a simple rearrangement of the definition:
$$H(1/q) = \frac{1}{q}\log_2(q) + \frac{q-1}{q}\log_2(q/(q-1)) = \log_2(q) - \frac{q-1}{q}\log_2(q-1).$$

To prove (2), start by writing
$$H(1/q) - H(1/q - \delta) = \int_{1/q - \delta}^{1/q} H'(z) dz = \int_{1/q - \delta}^{1/q} \log_2\left( \frac{1-z}{z} \right) dz.$$
Since $\log((1-z)/z)$ is convex, we can bound this above using the midpoint value by
$$\delta \cdot \log_2\left( \frac{1-(1/q + \delta/2)}{1/q + \delta/2} \right) dz = \delta \cdot \log_2(q-1) - \delta^2 \cdot \frac{q^2}{(q-1)\log_e(4)} + O(\delta^3),$$
where the last step is the Taylor expansion at $\delta=0$. 

Similarly, (3) follows by
$$H(1/q) - H(1/q - \delta) \leq \delta \cdot \log_2\left( \frac{1-(1/q - \delta/2)}{1/q - \delta/2} \right) dz = \delta \cdot \log_2(q-1) + \delta^2 \cdot \frac{q^2}{(q-1)\log_e(4)} + O(\delta^3).$$
\end{proof}

\subsection{Basic Tools for Rigidity and Kronecker Products}

We now give a number of basic tools which will be of use throughout our proofs.

\begin{proposition}[The mixed-product property] \label{mixedproductproperty}
Let $\F$ be any field, and let $A \in \F^{a_1 \times a_2}, B \in \F^{b_1 \times b_2}, C \in \F^{c_1 \times c_2}, D \in \F^{d_1 \times d_2}$ be any matrices over $\F$ with $a_2=c_1$ and $b_2=d_1$. Then, $(A \otimes B)\times(C \otimes D) = (A \times C) \otimes (B \times D)$.
\end{proposition}

\begin{proposition}
For any field $\F$, any positive integers $a,b$, and any matrices $A \in \F^{a \times a}$ and $B \in \F^{b \times b}$, we have $\rank(A \otimes B) = \rank(A) \cdot \rank(B)$.
\end{proposition}

\begin{proposition}\label{sumtoconcat}
For any field $\F$, integers $d_1, d_2, d_3, d_4$ and matrices $X_1 \in \F^{d_1 \times d_2}$, $X_2 \in \F^{d_2 \times d_3}$, $X_3 \in \F^{d_1 \times d_4}$, and $X_4 \in \F^{d_4 \times d_3}$, we have $$X_1 \times X_2 + X_3 \times X_4 = \left( X_1 | X_3 \right) \times \left( \frac{X_2}{X_4} \right),$$ where we are writing `$|$' to denote matrix concatenation.
\end{proposition}

\begin{lemma} \label{butterfly}
For any field $\F$, positive integers $q,n$, and matrices $M_1, \ldots, M_n \in \F^{q \times q}$, we have \begin{align}\label{butterflyeqn}\bigotimes_{i=1}^n M_i = \prod_{i=1}^n M_i^{\otimes \{ i \}} \otimes (I_{q^{n-1}})^{\otimes [n] \setminus \{ i \} }.\end{align}
\end{lemma}

\begin{proof}
We proceed by induction on $n$. The base case $n=1$ is true since then the right-hand side of \Cref{butterflyeqn} is simply equal to $M_1$. For the inductive step, we see that
\begin{align*}\bigotimes_{i=1}^n M_i &= \bigotimes_{i=1}^{n-1} M_i \otimes M_n \\ 
&= \left(\prod_{i=1}^{n-1} M_i^{\otimes \{ i \}} \otimes (I_{q^{n-2}})^{\otimes [n-1] \setminus \{ i \} } \right) \otimes M_n \\
&= \left( \left(\prod_{i=1}^{n-1} M_i^{\otimes \{ i \}} \otimes (I_{q^{n-2}})^{\otimes [n-1] \setminus \{ i \} } \right) \times I_{q^{n-1}} \right) \otimes (I_q \times M_n) \\
&= \left( \left(\prod_{i=1}^{n-1} M_i^{\otimes \{ i \}} \otimes (I_{q^{n-2}})^{\otimes [n-1] \setminus \{ i \} } \right) \otimes I_{2} \right) \times (I_{q^{n-1}} \otimes M_n) &&\text{ (by \Cref{mixedproductproperty})} \\
&=  \left(\prod_{i=1}^{n-1} \left( \left( M_i^{\otimes \{ i \}} \otimes (I_{q^{n-2}})^{\otimes [n-1] \setminus \{ i \} }   \right) \otimes I_{2} \right) \right) \times (I_{q^{n-1}} \otimes M_n) \\
&=  \left(\prod_{i=1}^{n-1} M_i^{\otimes \{ i \}} \otimes (I_{q^{n-1}})^{\otimes [n] \setminus \{ i \} }  \right) \times (I_{q^{n-1}} \otimes M_n) \\
&=  \prod_{i=1}^{n} M_i^{\otimes \{ i \}} \otimes (I_{q^{n-1}})^{\otimes [n] \setminus \{ i \} } ,
\end{align*}
as desired.
\end{proof}

\begin{definition}
For any field $\F$, positive integer $q$, and matrix $M \in \F^{q \times q}$, we say $M$ is an \emph{outer-1} matrix if, for all $i,j \in \{0,1,\ldots,q-1\}$ with $i=0$ or $j=0$ (or both) we have $M[i,j]=1$. We similarly say $M$ is an \emph{outer-0} matrix if we have $M[i,j]=0$ for all such $i,j$, and an \emph{outer-nonzero} matrix if we have $M[i,j]\neq 0$ for all such $i,j$.
\end{definition}

\begin{lemma} \label{reducedimsingle}
For any field $\F$, positive integer $q$, and outer-nonzero matrix $M \in \F^{q \times q}$, there are
\begin{itemize}
    \item an outer-1 matrix $M' \in \F^{q \times q}$, and
    \item two invertible diagonal matrices $D, D' \in \F^{q \times q}$,
\end{itemize}
such that $M = D \times M' \times D'$.
\end{lemma}

\begin{proof}
We first define the diagonal matrices $G,G' \in \F^{q \times q}$ by: For $i \in \{0,1,\ldots,q-1\}$, set $G[i,i] = 1/M[i,0]$ and $G'[i,i] = M[0,0]/M[0,i]$. These are well-defined and invertible since $M$ is an outer-nonzero matrix. Let $M' = G \times M \times G'$; we can see that for any $i \in \{0,1,\ldots,q-1\}$ we have $M'[i,0]=M[i,0] \cdot G[i,i] \cdot G'[0,0] = M[i,0] \cdot (1/M[i,0]) \cdot (M[0,0]/M[0,0]) = 1$, and for any $j \in \{0,1,\ldots,q-1\}$ we have $M'[0,j]=M[0,j] \cdot G[0,0] \cdot G'[j,j] = M[0,j] \cdot (1/M[0,0]) \cdot (M[0,0]/M[0,j]) = 1$, so $M'$ is an outer-1 matrix. Finally we can pick $D = G^{-1}$ and $D' = G'^{-1}$ so that $M = D \times M' \times D'$.
\end{proof}

\begin{lemma} \label{reducedim}
For any field $\F$, positive integers $n,q$, and outer-nonzero matrices $M_1, \ldots, M_n \in \F^{q \times q}$, there are
\begin{itemize}
    \item outer-1 matrices $M'_1, \ldots, M'_n \in \F^{q \times q}$, and
    \item two invertible diagonal matrices $D, D' \in \F^{q^n \times q^n}$,
\end{itemize}
such that $\bigotimes_{\ell=1}^n M_\ell = D \times \left( \bigotimes_{\ell=1}^n M'_\ell \right) \times D'$.
\end{lemma}

\begin{proof}
By \Cref{reducedimsingle}, for each $\ell \in [n]$, there are invertible diagonal matrices $D_\ell,D'_\ell \in \F^{q \times q}$ and an outer-1 matrix $M'_\ell \in \F^{q \times q}$ such that $M_\ell = D_\ell \times M'_\ell \times D'_\ell$. Then, by \Cref{mixedproductproperty},
$$\bigotimes_{\ell=1}^n M_\ell = \bigotimes_{\ell=1}^n (D_\ell \times M'_\ell \times D'_\ell) = \left( \bigotimes_{\ell=1}^n D_\ell \right) \times \left( \bigotimes_{\ell=1}^n M'_\ell \right) \times \left( \bigotimes_{\ell=1}^n D'_\ell \right).$$ We can thus pick $D = \bigotimes_{\ell=1}^n D_\ell$ and $D' = \bigotimes_{\ell=1}^n D'_\ell$ as desired.
\end{proof}

\begin{lemma} \label{rigiditydiagonal}
For any field $\F$, positive integers $q,r$, and matrices $A,B,D,D' \in \F^{q \times q}$ such that $D$ and $D'$ are invertible diagonal matrices with $A = D \times B \times D'$, we have that $\rig_A(r) = \rig_{B}(r)$.
\end{lemma}
\begin{proof}
By definition of $\rig_{B}(r)$, there are matrices $L,S \in \F^{q \times q}$ such that $\rank(L) \leq r$, $nnz(S) \leq \rig_{B}(r)$, and $B=L+S$. It follows that $A = D \times L \times D' + D \times S \times D'$. Since multiplying on the left or right by a full-rank diagonal matrix does not change the rank or sparsity of a matrix, this expression shows that $\rig_A(r) \leq \rig_{B}(r)$. A symmetric argument also shows that $\rig_A(r) \geq \rig_{B}(r)$ as desired.
\end{proof}

The next Lemma, which shows that the product of non-rigid matrices is also non-rigid, was also used by \cite[Lemma~2.18]{dvir2019fourier}.

\begin{lemma} \label{lem:productnotrigid}
For any field $\F$, positive integers $q,r$, and matrices $A,B,C,D \in \F^{q \times q}$ with $D$ a diagonal matrix and $C = A \times D \times B$, we have that $$\rig_C^{rc}(2r) \leq \rig_A^{rc}(r) \cdot \rig_B^{rc}(r).$$
\end{lemma}

\begin{proof}
Let $s_A := \rig_A^{rc}(r)$ and $s_B := \rig_B^{rc}(r)$. Write $A = L_A + S_A$ and $B = L_B + S_B$ where $L_A,L_B, S_A, S_B \in \F^{q \times q}$ are matrices with $\rank(L_A) \leq r$, $\rank(L_B) \leq r$, $\nnzr(S_A) \leq s_A$, $\nnzc(S_A) \leq s_A$, $\nnzr(S_B) \leq s_B$, and $\nnzc(S_B) \leq s_B$. We have that
$$C = (L_A + S_A) \times D \times (L_B + S_B) = L_A  \times D \times (L_B + S_B) + S_A \times D \times L_B + S_A \times D \times S_B.$$
The first two matrices in the right-hand-side, $L_A  \times D \times (L_B + S_B)$ and $S_A \times D \times L_B$, both have rank at most $r$, since $L_A$ and $L_B$ have rank at most $r$. The third, $M := S_A \times D \times S_B$, has both
$$\nnzr(M) \leq \nnzr(S_A) \cdot \nnzr(S_B),$$
$$\nnzc(M) \leq \nnzc(S_A) \cdot \nnzc(S_B).$$
It follows that \begin{align*}&\max\{\nnzr(M), \nnzc(M)\} \\ &\leq \max\{ \nnzr(S_A) \cdot \nnzr(S_B), \nnzc(S_A) \cdot \nnzc(S_B) \} \\ &\leq \max\{ \nnzr(S_A) , \nnzc(S_A)  \} \cdot \max\{ \nnzr(S_B), \nnzc(S_B) \} \\ &\leq s_A \cdot s_B.\end{align*}
This expression thus shows that $\rig_C^{rc}(2r) \leq s_A \cdot s_B$ as desired.
\end{proof}

\section{Framework for Designing Small Circuits from Non-Rigidity} \label{sec:framework}

We first note that an upper bound for a fixed matrix in a family of Kronecker products leads to one for the entire family.

\begin{lemma} \label{lem:fixedtobig}
For any field $\F$, fixed positive integers $q,t,d$, and matrix $M \in \F^{q \times q}$, suppose $M^{\otimes t} = \prod_{j=1}^d B_j$ for matrices $B_j$ for all $j \in [d]$ with $\nnz(B_j) = b_j$. Then, for all positive integers $n$ and $j \in [d]$ there are matrices $A_{n,j}$ with $\nnz(A_{n,j}) < b_j^{1 + n/t}$ and $M^{\otimes n} = \prod_{j=1}^d A_{n,j}$. If $t$ divides $n$, the upper bound can be further reduced to $\nnz(A_{n,j}) \leq b_j^{n/t}$.
\end{lemma}

\begin{proof}
Assuming $t$ divides $n$, we will show there are matrices $A_{n,j}$ with $\nnz(A_{n,j}) = b_j^{n/t}$ and $M^{\otimes n} = \prod_{j=1}^d A_{n,j}$. If $t$ does not divide $n$, we can instead apply this construction for the next multiple $n'>n$ of $t$, and then pick the appropriate submatrix of $M^{\otimes n'}$, to get $M^{\otimes n}$; we will thus have $\nnz(A_{n,j}) = b_j^{n'/t} < b_j^{1 + n/t}$.

Now, assuming $t$ divides $n$, then we can simply write $M^{\otimes n} = (\prod_{j=1}^d B_j)^{\otimes n/t} = \prod_{j=1}^d B_j^{\otimes n/t}$, and pick $A_{n,j} := B_j^{\otimes n/t}$, which has $\nnz(A_{n,j}) = \nnz(B_j^{\otimes n/t}) = \nnz(B_j)^{n/t} = b_j^{n/t}$, as desired.
\end{proof}

Next, we observe that rigidity upper bounds can be used to give depth-2 synchronous circuit upper bounds.

\begin{lemma} \label{rigtosparsub}
For any field $\F$, fixed positive integers $r,q$, and matrix $M \in \F^{q \times q}$, there are matrices $B \in \F^{q \times (q+r)}$ and $C \in \F^{(q+r) \times q}$ such that $M = B \times C$, $\nnz(B) = q \cdot r + \rig_M(r)$, and $\nnz(C) = q \cdot (r+1)$.
\end{lemma}

\begin{proof}
By definition of rigidity, we can write $M = L+S$ for matrices $L,S \in \F^{q \times q}$ with $\rank(L) = r$ and $\nnz(S) = \rig_M(r)$. In particular, there are matrices $B' \in \F^{q \times r}$ and $C' \in \F^{r \times q}$ such that $L = B' \times C'$. By \Cref{sumtoconcat}, our desired matrix decomposition is thus
$$M = \left( S | B' \right) \times \left( \frac{I_{q}}{C'} \right).$$
We have $\nnz(B) = \nnz(S) + \nnz(B') \leq \rig_M(r) + q \cdot r$, and $\nnz(C) = \nnz(I_q) + \nnz(C') \leq q + q\cdot r$.
\end{proof}

\begin{remark}\label{rigtosparsubrem}
Applying \Cref{rigtosparsub} to $M^T$ instead of $M$, we can alternatively obtain $B \in \F^{q \times (q+r)}$ and $C \in \F^{(q+r) \times q}$ such that $M = B \times C$, $\nnz(B) = q \cdot (r+1)$, and $\nnz(C) = q \cdot r + \rig_M(r)$. In other words, we can choose either $B$ or $C$ to have higher sparsity.
\end{remark}

Finally, we show how to `symmetrize' the construction of \Cref{rigtosparsub} to extend it to small circuits of any depth $d \geq 2$.

\begin{theorem}\label{thm:cktfromrigidity}
For any field $\F$, positive integers $r,q$, and matrix $M \in \F^{q \times q}$, let $$c:= \log_q((r+1) \cdot (r + \rig_M(r)/q)).$$
Then, for every positive integers $n,d$, setting $N = q^n$, the matrix $M^{\otimes n} \in \F^{N \times N}$ can be written as $M^{\otimes n} = \prod_{j=1}^d A_{n,j}$ for matrices $A_{n,j}$ with $\nnz(A_{n,j}) < q^{1-c} \cdot N^{1 + c/d}.$ If $d$ divides $n$, the upper bound can be further reduced to $\nnz(A_{n,j}) \leq N^{1 + c/d}.$
\end{theorem}

\begin{proof}
Using \Cref{rigtosparsub} and \Cref{rigtosparsubrem}, there are matrices $B,B',C,C'$ such that $M=B \times C = C' \times B'$, $\nnz(B)=\nnz(B')=q \cdot r + \rig_M(r)$, and $\nnz(C)=\nnz(C')=q \cdot(r+1)$. We thus have the following $d$ ways to write $M$ as a product of $d$ matrices:

\begin{align*}
    M &= B \times C \times I_q \times I_q \times I_q \times \cdots \times I_q \times I_q \times I_q \\
    M &= I_q \times B \times C \times I_q \times I_q \times \cdots \times I_q \times I_q \times I_q \\
    M &= I_q \times I_q \times B \times C \times I_q \times \cdots \times I_q \times I_q \times I_q \\
    M &= I_q \times I_q \times I_q \times B \times C \times \cdots \times I_q \times I_q \times I_q \\
    &\vdots \\
    M &= I_q \times I_q \times I_q \times I_q \times I_q \times \cdots \times B \times C \times I_q \\
    M &= I_q \times I_q \times I_q \times I_q \times I_q \times \cdots \times I_q \times B \times C \\
    M &= C' \times I_q \times I_q \times I_q \times I_q \times \cdots \times I_q \times I_q \times B'. \\
\end{align*}

Applying \Cref{mixedproductproperty}, there are thus permutation matrices $P_j, P'_j$ for each $j \in [d]$ such that we can write $M^{\otimes d}$ as:
$$M^{\otimes d} = \left( P_1 \times (B \otimes C' \otimes I_{q^{d-2}}) \times P'_1 \right) \times \left( \prod_{j=2}^{d-1} P_j \times \left( B \otimes C \otimes I_{q^{d-2}} \right) \times P'_j  \right) \times \left( P_d \times (B' \otimes C \otimes I_{q^{d-2}}) \times P'_d \right).$$
Since $\nnz(B)=\nnz(B')$ and $\nnz(C)=\nnz(C')$, this is expressing $M^{\otimes d}$ as a product of $d$ matrices, each of which has sparsity
$$\nnz(B \otimes C \otimes I_{q^{d-2}}) = \nnz(B) \cdot \nnz(C) \cdot \nnz(I_{q^{d-2}}) = (q \cdot r + \rig_M(r)) \cdot (q \cdot(r+1)) \cdot q^{d-2}.$$
Assume first that $d$ divides $n$. Applying \Cref{lem:fixedtobig}, it follows that the matrix $M^{\otimes n}$ can be written as $M^{\otimes n} = \prod_{j=1}^d A_{n,j}$ for matrices $A_{n,j}$ with 
\begin{align*}\nnz(A_{n,j}) &\leq ((q \cdot r + \rig_M(r)) \cdot (q \cdot(r+1)) \cdot q^{d-2})^{n/d} \\
&= q^{n} \cdot (r + \rig_M(r)/q)^{n/d} \cdot (r+1)^{n/d} \\
&= q^{n \cdot \left(1 + \frac{\log_q((r+1) \cdot (r + \rig_M(r)/q))}{d} \right)}\\
&= N^{(1 + \frac{c}{d})},
\end{align*}
where $N=q^n$ so that $M^{\otimes n} \in \F^{N \times N}$, and $c:= \log_q((r+1) \cdot (r + \rig_M(r)/q))$, as desired.

Next, consider when $d$ does not divide $n$. Let $n'$ be the largest integer less than $n$ such that $d$ divides $n'$, and let $k = n-n'$ so $k<d$. By the above argument, there are matrices $A_{n',1}, \ldots, A_{n',d}$ such that $M^{\otimes n'} = \prod_{j=1}^d A_{n',j}$ and $\nnz(A_{n',j}) \leq q^{n' \cdot (1 + c/d)}.$ For each $1 \leq \ell \leq k$ we can also write $M = \prod_{j=1}^d ([j=\ell] ~?~ M : I_q)$. Combining these $k+1$ expressions together, again using \Cref{mixedproductproperty}, it follows that there are permutation matrices $P_j, P'_j$ for each $j \in [d]$ such that
$$M^{\otimes n} = \left( \prod_{j=1}^{k} P_j \times \left( A_{n',j} \otimes M \otimes I_{q^{k-1}} \right) \times P'_j  \right) \times \left( \prod_{j=k+1}^{d} P_j \times \left( A_{n',j} \otimes I_{q^{k}} \right) \times P'_j  \right).$$
We can calculate that $\nnz(A_{n',j} \otimes M \otimes I_{q^{k-1}}) \leq q^{n' \cdot (1+c/d) + k + 1} < q^{(1-c) + n\cdot(1+c/d)}$, and similarly $\nnz(A_{n',j} \otimes I_{q^{k}}) < q^{(1-c) + n\cdot(1+c/d)}$, which concludes the proof like before.
\end{proof}

In the proof of \Cref{thm:cktfromrigidity}, we made use of \Cref{rigtosparsubrem} that our fixed upper bound from non-rigidity can be made symmetric. For fixed upper bounds designed in other ways, this may not be the case. Below in \Cref{sec:moregeneral}, we will nonetheless show that any nontrivial fixed upper bound can be used to prove a result similar to \Cref{thm:cktfromrigidity}. For now, in this section and the next, we will focus specifically on our upper bounds from non-rigidity.

\subsection{Slightly Smaller Circuits with Larger Depth}

In this subsection, we remark that we can remove the $q^{1-c}$ factor from the circuit size in \Cref{thm:cktfromrigidity} in exchange for a slight increase in depth (but not total size):

\begin{corollary}\label{cor:cktfromrigidity}
For any field $\F$, positive integers $r,q$, and matrix $M \in \F^{q \times q}$, let $$c:= \log_q((r+1) \cdot (r + \rig_M(r)/q)).$$
Then, for every positive integers $n,d$, with $d < o(n)$, setting $N = q^n$, the matrix $M^{\otimes n} \in \F^{N \times N}$ has a synchronous linear circuit of size $(1 + o(1)) \cdot d \cdot q^{n \cdot (1 + c/d)}.$ 
\end{corollary}

\begin{proof}
Let $n'$ be the integer in the range $n \geq n' > n-d$ such that $d$ divides $n'$, and let $k = n-n'$. Applying \Cref{thm:cktfromrigidity} to $M^{\otimes n'}$, we see that it has a synchronous circuit of size $d \cdot q^{n' \cdot (1 + c/d)}$. Thus, $M^{\otimes n'} \otimes I_{q^k}$ has a synchronous circuit of size $d \cdot q^{n' \cdot (1 + c/d)} \cdot q^k = d \cdot q^{n \cdot (1 + c/d)} / q^{k \cdot c/d}$. Next, again by applying \Cref{thm:cktfromrigidity}, but this time for depth $k$, we see that $M^{\otimes k}$ has a synchronous circuit of size $k \cdot q^{k+c}$, and so $I_{q^{n'}} \otimes M^{\otimes k}$ has a synchronous circuit of size $q^{n'} \cdot k \cdot q^{k+c} = k \cdot q^{n+c}$. Hence, since $M^{\otimes n} = M^{\otimes n'} \otimes M^{\otimes k} = (M^{\otimes n'} \otimes I_{q^k}) \times (I_{q^{n'}} \otimes M^{\otimes k})$, it follows that $M^{\otimes n}$ has a synchronous circuit of size 
$$d \cdot q^{n \cdot (1 + c/d)} / q^{k \cdot c/d} + k \cdot q^{n+c} = q^{n \cdot (1 + c/d)} \cdot \left( \frac{d}{q^{kc/d}} + \frac{k}{q^{c(n/d - 1)}} \right) \leq (1 + o(1)) \cdot d \cdot q^{n \cdot (1 + c/d)}.$$
\end{proof}

\begin{corollary}\label{cor:cktfromrigidity2}
For any field $\F$, positive integers $r,q$, and matrix $M \in \F^{q \times q}$, let $$c:= \log_q((r+1) \cdot (r + \rig_M(r)/q)).$$
Then, for every positive integer $n$, setting $N = q^n$, the matrix $M^{\otimes n} \in \F^{N \times N}$ has a synchronous linear circuit of size $(c \cdot e \cdot \log_e(2) + o(1)) \cdot N \cdot \log_2(N).$ 
\end{corollary}

\begin{proof}
We will apply \Cref{cor:cktfromrigidity} with $d = c \cdot \log_e(N)$. The resulting circuit size is
$$(1 + o(1)) \cdot d \cdot q^{n \cdot (1 + c/d)} = (1 + o(1)) \cdot c \cdot \log_e(N) \cdot N \cdot e = (c \cdot e \cdot \log_e(2) + o(1)) \cdot N \cdot \log_2(N).$$
\end{proof}

\section{Smaller Circuits from Rank-1 Rigidity}

In this section, we study the rank-1 rigidities of a number of families of matrices. We will find that many matrices of interest have fairly low rank-1 rigidity. These constructions can be combined with the results of the previous section to prove our main results.

\subsection{Kronecker Power Matrices}

\begin{lemma} \label{mcubed1}
For any field $\F$ and any outer-1 matrix $M \in \F^{2 \times 2}$, we have $\rig_{M^{\otimes 3}}(1) \leq 23$.
\end{lemma}

\begin{proof}
Since $M$ is an outer-1 matrix, there is an $\omega \in \F$ such that $$M = 
\begin{bmatrix}
1  & 1 \\
1  & \omega
\end{bmatrix}.$$
We can index entries of $M^{\otimes 3}$ by vectors $x,y \in \{0,1\}^3$, so that $M^{\otimes 3}[x,y] = \omega^{\langle x,y \rangle_{\Z}}$. Consider the matrix $L \in \F^{8 \times 8}$ given by 
$$L[x,y] = \begin{cases} 
    \omega^{-1} &\text{ if } x=y=(0,0,0), \\
    1 &\text{ if } x=(0,0,0) \text{ and } y \neq (0,0,0), \\
    1 &\text{ if } x \neq (0,0,0) \text{ and } y = (0,0,0), \\
    \omega &\text{ if } x \neq (0,0,0) \text{ and } y \neq (0,0,0). \\
\end{cases}$$
$L$ has rank $1$, and we can see that $L[x,y]=M^{\otimes 3}[x,y]$ unless:
\begin{itemize}
    \item $x=y=(0,0,0)$, or
    \item $x \neq (0,0,0)$, $y \neq (0,0,0)$, and $\langle x,y \rangle_{\Z} \neq 1$.
\end{itemize}
We can count that:
\begin{itemize}
    \item When $x=(1,0,0)$, $x=(0,1,0)$, or $x=(0,0,1)$, there are 3 choices of $y \neq (0,0,0)$ with $\langle x,y \rangle_\Z = 0$.
    \item When $x=(1,1,0)$, $x=(0,1,1)$, or $x=(1,0,1)$, there is 1 choice of $y \neq (0,0,0)$ with $\langle x,y \rangle_\Z = 0$, and 2 choices with $\langle x,y \rangle_\Z = 2$.
    \item When $x=(1,1,1)$, there are 3 choices of $y \neq (0,0,0)$ with $\langle x,y \rangle_\Z = 2$, and 1 choice with $\langle x,y \rangle_\Z = 3$.
\end{itemize}
Overall, $L$ and $M^{\otimes 3}$ differ in $1\cdot 1 + 3 \cdot 3 + 3 \cdot 3 + 1 \cdot 4 = 23$ entries.
\end{proof}

\begin{lemma} \label{mcubed}
For any field $\F$ and any matrix $M \in \F^{2 \times 2}$, we have $\rig_{M^{\otimes 3}}(1) \leq 23$.
\end{lemma}

\begin{proof}
By \Cref{reducedim} and \Cref{rigiditydiagonal}, it is sufficient to consider the case when $M$ is an outer-1 matrix. The result then follows from \Cref{mcubed1}.
\end{proof}

\begin{theorem}
For any field $\F$, matrix $M \in \F^{2 \times 2}$, and positive integers $d,n>1$, the matrix $M^{\otimes n} \in \F^{N \times N}$ for $N = 2^n$ has a depth-$d$ linear circuit of size $2^{\eps} \cdot N^{1 + (1-\eps)/d}$ for some constant $\eps > 0.01526$.
\end{theorem}

\begin{proof}
Applying \Cref{thm:cktfromrigidity} with $M^{\otimes 3}$, $q=8$, and $r=1$, combined with the rigidity bound of \Cref{mcubed}, shows that $M^{\otimes n}$ has a depth-$d$ linear circuit of size $2^{1-c} \cdot N^{1 + c/d}$ for
$$c = \log_q\left( (r+1) \cdot \left(r + \frac{\rig_M(r)}{q} \right) \right) \leq \log_8\left( 2 \cdot \left(1 + \frac{23}{8} \right) \right) < 0.98474 = 1-\eps.$$
\end{proof}

\begin{corollary}
For any field $\F$, matrix $M \in \F^{2 \times 2}$, and positive integer $n>1$, the matrix $M^{\otimes n} \in \F^{N \times N}$ for $N = 2^n$ has a synchronous linear circuit of size $((1-\eps) \cdot e \log_e(2) + o(1)) \cdot N \log_2 N$ for some constant $\eps > 0.01526$.
\end{corollary}

\begin{proof}
Apply \Cref{cor:cktfromrigidity2} with the same rigidity bound of \Cref{mcubed}.
\end{proof}

\subsection{Walsh-Hadamard Transform}

\begin{lemma} \label{had2rig}
Over any field $\F$ with $\ch(\F) \neq 2$, we have $\rig_{H_2}(1) = 4$.
\end{lemma}

\begin{proof} First, to see that $\rig_{H_2}(1) \leq 4$, we can verify that
\begin{align*}H_2 = 
\begin{bmatrix}
1  & 1  & 1  & 1   \\
1 & -1  & 1 & -1  \\
1 & 1 & -1  & -1   \\
1  & -1 & -1 & 1 
\end{bmatrix}
= 
\begin{bmatrix}
-1  & 1  & 1  & 1   \\
1 & -1  & -1 & -1  \\
1 & -1 & -1  & -1   \\
1  & -1 & -1 & -1 
\end{bmatrix} + \begin{bmatrix}
2 & 0 & 0 & 0   \\
0 & 0 & 2 & 0  \\
0 & 2 & 0 & 0   \\
0 & 0 & 0 & 2 
\end{bmatrix}.
\end{align*}
This is the sum of a rank-1 matrix (where each row after the first is the negation of the first row), and a matrix with 4 nonzero entries, as desired.

The bound $\rig_{H_2}(1) \geq 4$ actually follows from the known general lower bound $\rig_{H_n}(r) \geq 2^{2n-2}/r$~\cite{midrijanis2005three,de2006lower}, but we prove it here for completeness using the simple proof strategy of~\cite{midrijanis2005three}. Recall that we can write $H_2$ as a block matrix as
\begin{align*}H_2 = 
\begin{bmatrix}
H_1 & H_1 \\
H_1 & -H_1
\end{bmatrix}.
\end{align*}
Each copy of $H_1$ has rank $2$, so we must change at least one entry in each $H_1$ to drop the rank of the whole matrix to $1$. Since there are four disjoint copies, we must change at least four entries.
\end{proof}

\begin{lemma} \label{had3rig}
Over any field $\F$, we have $\rig_{H_3}(1) \leq 22$.
\end{lemma}

\begin{proof}
We use the same construction as in \Cref{mcubed}, with $\omega=-1$ so that $M^{\otimes 3} = H_3$. In this case, there is one more correct entry than in the general case, since when $x=y=(1,1,1)$, we have $M^{\otimes 3}[x,y]=\omega^3$ and $L[x,y]=\omega$, but these are equal when $\omega=-1$, so the number of errors is only $23-1=22$.
\end{proof}

\begin{lemma}\label{had4rig}
Over any field $\F$, we have $\rig_{H_4}(1) \leq 96$.
\end{lemma}

\begin{proof}
In the proof of \Cref{had2rig}, we showed there is a matrix $A \in \{-1,1\}^{4 \times 4}$ which differs from $H_2$ in $4$ entries, and which has rank $1$ over any field. Let $B = A^{\otimes 2} \in \{-1,1\}^{16 \times 16}$. We have that $\rank(B) = \rank(A)^2 = 1$. Indexing the rows and columns of $H_2$ by $\{0,1,2,3\}$, and the rows and columns of $H_4$ by $\{0,1,2,3\}^2$, we see that for $a,b,c,d \in \{0,1,2,3\}$ we have
$$\frac{B[(a,b),(c,d)]}{H_4[(a,b),(c,d)]} = \frac{A[a,c] \cdot A[b,d]}{H_2[a,c] \cdot H_2[b,d]}.$$
This will equal $1$ (and hence the $[(a,b),(c,d)]$ entries of $B$ and $H_4$ will be equal) whenever either:
\begin{itemize}
    \item $A[a,c]=H_2[a,c]$ and $A[b,d]=H_2[b,d]$, which happens for $(16-4)^2 = 144$ values of $a,b,c,d \in \{0,1,2,3\}$, or
    \item $A[a,c]\neq H_2[a,c]$ and $A[b,d]\neq H_2[b,d]$ (since all these values are in $\{-1,1\}$), which happens for $4^2 = 16$ values of $a,b,c,d \in \{0,1,2,3\}$.
\end{itemize}
Thus, $B$ only differs from $H_4$ in $16^2 - 144 - 16 = 96$ entries, as desired.
\end{proof}

\begin{remark}
I verified using a brute-force search that \Cref{had3rig} and \Cref{had4rig} are tight over any field $\F$ with $\ch(\F) \neq 2$. I unfortunately haven't found more enlightening proofs of these facts.
\end{remark}

\begin{theorem}
For any field $\F$ and positive integers $d,n>1$, the matrix $H_n \in \F^{N \times N}$ for $N = 2^n$ has a depth-$d$ linear circuit of size $\leq 2^{\eps} \cdot N^{1 + (1-\eps)/d + O(d/n)}$ for some constant $\eps > 0.04816$.
\end{theorem}

\begin{proof}
Applying \Cref{thm:cktfromrigidity} with $H_4$, $q=16$, and $r=1$, combined with the rigidity bound of \Cref{had4rig}, shows that $H_n=H_1^{\otimes n}$ has a depth-$d$ linear circuit of size $2^{1-c} \cdot N^{1 + c/d}$ for
$$c = \log_q\left( (r+1) \cdot \left(r + \frac{\rig_M(r)}{q} \right) \right) \leq \log_{16}\left( 2 \cdot \left(1 + \frac{96}{16} \right) \right) < 0.95184 = 1-\eps.$$
\end{proof}

\begin{corollary}
For any field $\F$ and positive integer $n>1$, the matrix $H_n \in \F^{N \times N}$ for $N = 2^n$ has a synchronous linear circuit of size $((1-\eps) \cdot e \log_e(2) + o(1)) \cdot N \log_2 N$ for some constant $\eps > 0.04816$.
\end{corollary}

\begin{proof}
Apply \Cref{cor:cktfromrigidity2} with the same rigidity bound of \Cref{had4rig}.
\end{proof}

\subsection{Fourier Transform}

In order to use the approach of \Cref{thm:cktfromrigidity} to prove that the $N \times N$ Fourier transform $F_N$ has depth-$d$ circuits of size $O(N^{1 + c/d})$ for some $c<1$, we would need it to be the case that, for some positive integers $N>r>0$, we have
$$\log_N((r+1) \cdot (r + \rig_{F_N}(r)/N)) < 1.$$
We next remark that known rigidity lower bounds for $F_N$ show that this is never the case. In fact, the proof extends to any Vandermonde matrix.

\begin{proposition}
For any positive integers $N>r\geq 0$, the $N \times N$ Fourier transform matrix $F_N$ has$$(r+1) \cdot (r + \rig_{F_N}(r)/N) \geq N.$$
\end{proposition}

\begin{proof} Shparlinski~\cite{Shparlinski} shows that $\rig_{F_N}(r) \geq (N-r)^2/(r+1)$; for completeness, we prove this below in \Cref{ShparlinskiLB}. It then follows that:
\begin{align*}
    (r+1) \cdot \left(r + \frac{\rig_{F_N}(r)}{N}\right) &\geq (r+1) \cdot \left(r + \frac{(N-r)^2}{(r+1)\cdot N}\right) \\
    &= \frac{1}{N}\left(N^2 + r^2 + Nr(r-1) \right) \\
    &\geq\frac{1}{N}\left(N^2\right) \\
    &= N.\qedhere
\end{align*}
\end{proof}

We next prove a Lemma which we will need in the proof of Shparlinski's rigidity lower bound.
\begin{lemma} \label{fourierriglemma}
For any positive integers $N>r \geq 0$, any integer $0 \leq k < n-r$, and any $S \subseteq [n]_0$ of size $|S|=r$, let $M_{k,S}$ be the $r \times r$ submatrix of $F_N$ consisting of the rows of $\{k, k+1, k+2, \ldots, k+r-1\}$ and the columns of $S$. Then, $M_{k,S}$ has full rank.
\end{lemma}

\begin{proof}
Indexing the rows of $M_{k,s}$ by $[r]_0$ and the columns by $S$, we have for $j \in [r]_0$ and $s \in S$ that $M_{k,S}[j,s] = \omega_N^{j \cdot s} = (\omega_N^{s})^j$, where $\omega_N = e^{2 i \cdot \pi / N} \in \C$ is a primitive $N$th root of unity. Assume to the contrary that $M_{k,S}$ does not have full rank. Thus, there is a nontrivial linear combination of its rows summing to zero. This means that there are $a_0, a_1, \ldots, a_{r-1} \in \C$, which are not all $0$, such that, for each $s \in S$, we have $$\sum_{j=0}^{r-1} a_j \cdot (\omega_N^{s})^j = 0.$$ In other words, the $r$ different values $\{ \omega_N^s \mid s \in S\}$ are all roots of the polynomial $p(z) = \sum_{j=0}^{r-1} a_j \cdot z^j$. However, $p$ is a nonzero polynomial of degree at most $r-1$, so it cannot have $r$ roots, a contradiction.
\end{proof}

\begin{lemma}[\cite{Shparlinski}] \label{ShparlinskiLB}
For any positive integers $N > r \geq 0$, we have $\rig_{F_N}(r) \geq (N-r)^2/(r+1)$.
\end{lemma}

\begin{proof}
Suppose that one can change $t$ entries of $F_N$ to make its rank at most $r$. For $k \in [N-r]_0$, let $t_k$ be the number of changes which are in rows $\{k, k+1, k+2, \ldots, k+r\}$. Since each change contributes to at most $r+1$ of the $t_k$ values, we have that $\sum_{k=0}^{N-r-1} T_k \leq (r+1) \cdot t$. Thus, by the pigeonhole principle, there must be a $k^* \in [N-r]_0$ such that $t_{k^*} \leq (r+1) \cdot t / (N-r)$. Let $S \subseteq [N]_0$ be the columns of $F_N$ such that none of the changes in rows $\{k^*, k^*+1, k^*+2, \ldots, k^*+r\}$ is in a column of $S$. It must be that $|S| \leq r$, since otherwise, by \Cref{fourierriglemma}, the matrix $M_{k^*,S}$ has rank $r+1$ and we did not make any changes to it. On the other hand, by definition, $|S| \geq N - t_{k^*} \geq N - (r+1) \cdot t / (N-r)$. It follows that $r \geq N - (r+1) \cdot t / (N-r)$, which rearranges to the desired $t \geq (N-r)^2/(r+1)$.
\end{proof}

\subsection{Disjointness} \label{sec:disjconst}

Recall the Disjointness marix $R_n \in \F^{N \times N}$ from \Cref{subsubsec:matricesdefined}. The approach of \Cref{thm:cktfromrigidity} can be used to prove that $R_n$ has depth-$d$ linear circuits of size $N^{1 + (1-\eps)/d}$. However, since $R_n$ is very sparse (it has $\nnz(R_n) = 3^n \leq N^{1.585}$) it is almost immediate that it has depth-$d$ circuits of size $O(N^{1 + c/d})$ for $c = \log_2(1.5) < 0.585$. In fact, using a construction of Jukna and Sergeev~\cite{jukna2013complexity}, we can do even better than this, improving to $c<0.5432$. We give the construction in the remainder of this section.

\begin{lemma}[{\cite[Lemma~4.2]{jukna2013complexity}}]
Let $t = \log_2(1 + \sqrt{2}) < 1.28$. For any field $\F$ and positive integer $n$, there are matrices $A_n, B_n \in \F^{2^n \times 2^n}$ with $\nnz(A_n), \nnz(B_n) \leq O(2^{t \cdot n})$ such that $R_n = A_n \times B_n$. 
\end{lemma}

\begin{proof}
We show how to partition the $1$s of $R_n$ into squares (all-1s combinatorial rectangles with the same number of rows and columns) and rectangles (all-1s combinatorial rectangles with \emph{twice as many} rows as columns). Our partition is defined recursively. Let $s_n$ be the sum of the side-lengths of the squares in the partition of $R_n$, and let $r_n$ be the sum of the shorter side-lengths of the rectangles. For $$R_1 := \begin{bmatrix}
1  & 1    \\
1  & 0 
\end{bmatrix},$$
we can see that $s_1=r_1=1$. Next, from the recursive definition $$R_n := \begin{bmatrix}
R_{n-1}  & R_{n-1}    \\
R_{n-1}  & 0 
\end{bmatrix},$$
we see that the three copies of any $s \times s$ square in $R_{n-1}$ can be partitioned into a $s \times s$ square and a $2s \times s$ rectangle in $R_n$, and the three copies of any $2s \times s$ rectangle in $R_{n-1}$ can be partitioned into a $2s \times s$ rectangle and a $2s \times 2s$ square in $R_n$. It follows that we get the recurrence
$$\begin{bmatrix}
s_n    \\
r_n 
\end{bmatrix} = \begin{bmatrix}
1  & 2    \\
1  & 1 
\end{bmatrix} \times \begin{bmatrix}
s_{n-1}    \\
r_{n-1} 
\end{bmatrix}.$$ Since the matrix $\begin{bmatrix}
1  & 2    \\
1  & 1 
\end{bmatrix}$ has eigenvalues $1 \pm \sqrt{2}$, it follows that $s_n, r_n \leq O((1 + \sqrt{2})^n)$. We have thus written the $1$s of $R_n$ as a disjoint sum of combinatorial rectangles whose side-lengths sum to $O((1 + \sqrt{2})^n) = O(2^{t \cdot n})$, from which the result follows.
\end{proof}

Following the same construction as \Cref{thm:cktfromrigidity}, we get:

\begin{proposition}
For any field $\F$ and any positive integers $n,d$, let $N = 2^n$ and let $c = 2(\log_2(1+\sqrt{2}) - 1) < 0.5432$. There are $d$ matrices $A_{n,1}, \ldots, A_{n,d}$ such that $R_n = \prod_{j=1}^d A_{n,j}$ and $\nnz(A_{n,j}) \leq O(N^{1 + c/d})$ for all $j \in [d]$.
\end{proposition}

\section{Rigidity of Disjointness}

Recall that $R_1 := \begin{bmatrix}
1  & 1    \\
1  & 0 
\end{bmatrix}$ and $R_n := R^{\otimes n}$. For $x,y \in \{0,1\}^n$, we can equivalently define:
$$R_n[x,y] = \begin{cases} 0 &\text{ if there is an $\ell\in [n]_0$ such that $x[\ell]=y[\ell]=1$,} \\ 1 &\text{ otherwise.}  \end{cases}$$

For positive integers $k\leq n$, write $\binom{n}{<k} := \sum_{i=0}^{k-1} \binom{n}{i}$ and $\binom{n}{\leq k} := \sum_{i=0}^{k} \binom{n}{i}$. By standard bounds, we have that $$\binom{n}{< k} \leq \binom{n}{\leq k} \leq 2^{n},$$ and if $k \leq n/2$, then $$\binom{n}{< k} \leq \binom{n}{\leq k} \leq 2^{n \cdot H(n/k)},$$ where $H(p)$ is the binary entropy function.

\begin{lemma} \label{lem:fewrows}
For any positive integers $k \leq n$, we can remove $\binom{n}{< k}$ rows and $\binom{n}{< k}$ columns of $R_n$, so that the number of nonzero entries in any row or column of the resulting matrix is at most $\binom{n-k}{\leq n-2k}$.
\end{lemma}

\begin{proof}
Our construction is as follows. We remove the rows corresponding to $x \in \{0,1\}^k$ with $|x| < k$, so that the number we remove is indeed $\binom{n}{<k}$. We similarly remove the $\binom{n}{<k}$ columns corresponding to $y \in \{0,1\}^k$ with $|y| < k$.

Now, consider any $x \in \{0,1\}^n$ corresponding to a row we have not removed. Thus, $|x| \geq k$. For a given $y \in \{0,1\}^n$ which corresponds to a column we have not removed (and hence with $|y| \geq k$), we have $R_n[x,y]=1$ if and only if there is no $\ell \in [n]_0$ such that $x[\ell]=y[\ell]=1$. In other words, defining $S_x := \{\ell \in [n]_0 \mid x[\ell]=1\}$ and $S_y := \{\ell \in [n]_0 \mid y[\ell]=1\}$ (which are in bijection with $x$ and $y$), we have that $|S_y| \geq k$ and $S_y \subseteq [n]_0 \setminus S_x$. Since $|[n]_0 \setminus S_x| = n - |x| \geq n-k$, the number of choices for $y$ is hence at most the number of ways to choose a set $S_y$ of size at least $k$ from a set $[n]_0 \setminus S_x$ of size at most $n-k$, which is $$\sum_{j=k}^{n-k} \binom{n-k}{j} = \sum_{j=0}^{n-2k} \binom{n-k}{j+k} = \sum_{j=0}^{n-2k} \binom{n-k}{n - 2k - j} = \binom{n-k}{\leq n-2k}.$$ The bound on the number of nonzero entries in a column is identical.
\end{proof}

\begin{theorem} \label{thm:rnotrigid}
For any field $\F$, positive integer $n$, and $a \in (0,1)$, we have $\rig_{R_n}^{rc}\left(2 \cdot \binom{n}{<an}\right) \leq \binom{(1-a)n}{\leq (1-2a)n}$ over $\F$. In particular:
\begin{itemize}
    \item For any $\eps > 0$ we have $\rig_{R_n}^{rc}(2^{(\eps \log_2(1/\eps) + O(\eps)) \cdot n}) \leq 2^{(1-\eps) \cdot n}$, and
    \item For sufficiently small $\eps>0$ we have $\rig_{R_n}^{rc}(2^{(1 - \Theta(\eps^2 / \log^2(1/\eps)) \cdot n}) \leq 2^{\eps \cdot n}$, and
    \item We have $\rig_{R_n}^{rc}(O(2^{0.981 \cdot n})) < o( 2^{n/2} )$.
\end{itemize}
\end{theorem}

\begin{proof}
This follows from setting $k = a \cdot n$ in \Cref{lem:fewrows}, since setting one row or column of a matrix to zero is a rank-one update. For the particular parameter settings:

To see that $\rig_{R_n}^{rc}(2^{(\eps \log_2(1/\eps) + O(\eps)) \cdot n}) \leq 2^{(1-\eps) \cdot n}$, pick $a = \eps$. In that case, $2 \cdot \binom{n}{<\eps n} \leq 2^{H(\eps) \cdot n} \cdot \poly(n) \leq 2^{(\eps \log_2(1/\eps) + O(\eps)) \cdot n}$, and $\binom{(1-\eps)n}{\leq (1-2\eps)n} \leq 2^{(1-\eps)n}$.

To see that $\rig_{R_n}^{rc}(2^{(1 - \Theta(\eps^2 / \log^2(1/\eps)) \cdot n}) \leq 2^{\eps \cdot n}$, pick $a = 1/2 - \delta$ for an appropriate $\delta>0$ we will determine shortly. In that case, $2 \cdot \binom{n}{<(1/2 - \delta) \cdot n} \leq 2^{H(1/2 - \delta) \cdot n} \leq 2^{(1 - \Theta(\delta^2)) \cdot n}$, and $\binom{(1/2+\delta)n}{\leq 2\delta n} \leq 2^{(1/2 + \delta) \cdot H(4\delta/(1+2\delta)) \cdot n} \leq 2^{\Theta(\delta \log(1/\delta)) \cdot n}$. The result follows by picking $\delta$ such that the quantity $\Theta(\delta \log(1/\delta))$ in the sparsity bound is equal to $\eps$. In that case, $\delta^2 = \Theta(\eps^2 / \log^2(1/\eps))$.

To see that $\rig_{R_n}^{rc}(2^{0.89 \cdot n}) \leq \eps \cdot 2^{n/2}$, let $a^* \approx 0.4178$ be the larger solution in $[0,1/2]$ to $(1-a) \cdot H((1-2a)/(1-a)) = 1/2$. Then, for any $a>a^*$ it follows that $\binom{(1-a)n}{\leq (1-2a)n} < o( 2^{n/2})$, and $\binom{n}{<an} \leq 2^{H(a) \cdot n} \leq O(2^{H(a^*) \cdot n}) < O(2^{0.981 \cdot n})$.
\end{proof}

\section{Expressing Other Matrices In Terms Of Disjointness}

\begin{definition}
For any field $\F$, positive integer $n$, and function $f : \{0,1\}^n \to \F$, let $V_f \in \F^{2^n \times 2^n}$ denote the matrix which is given by, for $x,y \in \{0,1\}^n$, $V_f[x,y] := f(x \vee y)$, where `$x \vee y$' denotes the bit-wise OR of $x$ and $y$.
\end{definition}

\begin{definition}
For any field $\F$, positive integer $n$, and function $f : \{0,1\}^n \to \F$, let $a_f \in \F^{2^n}$ denote the vector with, for $z \in \{0,1\}^n$, the entry $a_f[z] := f(z)$. Let $b_f \in \F^{2^n}$ be the vector $b_f := R_n^{-1} \times a_f$. Let $D_f \in \F^{2^n \times 2^n}$ be the diagonal matrix of the entries of $b_f$, meaning for $z \in \{0,1\}^n$, we have $D_f[z,z] := b_f[z]$.
\end{definition}

\begin{lemma} \label{diagexpression}
For any field $\F$, positive integer $n$, and function $f : \{0,1\}^n \to \F$, we have $$V_f = R_n \times D_f \times R_n.$$
\end{lemma}

\begin{proof}
Recall that for $x,y \in \{0,1\}^n$, $$R_n[x,y] = \begin{cases} 1 &\text{ if } \langle x,y \rangle_{\Z} = 0, \\ 0 &\text{ otherwise.}  \end{cases}$$ It follows that, for any $x,y \in \{0,1\}^n$:
\begin{align*}
    (R_n \times D_f \times R_n)[x,y] 
    &= \sum_{z \in \{0,1\}^n}  R_n[x,z] \cdot D_f[z,z] \cdot R_n[z,y]\\
    &= \sum_{\overset{z \in \{0,1\}^n}{\langle x,z \rangle_{\Z} = \langle z,y \rangle_{\Z} = 0}}  D_f[z,z] \\
    &= \sum_{\overset{z \in \{0,1\}^n}{\langle (x \vee y),z \rangle_{\Z} = 0}}  D_f[z,z] \\
    &= \sum_{\overset{z \in \{0,1\}^n}{\langle (x \vee y),z \rangle_{\Z} = 0}}  b_f[z] \\
    &= \sum_{z \in \{0,1\}^n}  R_n[(x \vee y),z] \cdot b_f[z] \\
    &= (R_n \times b_f)[(x \vee y)] \\
    &= a_f[(x \vee y)] \\
    &= f(x \vee y) \\
    &= V_f[x,y], \\
\end{align*}
as desired.
\end{proof}

\begin{lemma} \label{lem:outer1tofunct}
For any field $\F$, positive integer $n$, and outer-1 matrices $M_1, \ldots M_n \in \F^{2 \times 2}$, there is a function $f : \{0,1\}^n \to \F$ and permutation matrices $\Pi_n, \Pi'_n \in \F^{2^n \times 2^n}$ such that $$\bigotimes_{i=1}^n M_i = \Pi_n \times V_f \times \Pi'_n.$$
\end{lemma}

\begin{proof}
For each $i \in [n]$, let $\omega_i \in \F$ be the element such that $$M_i = \begin{bmatrix}
1 & 1 \\
1 & \omega_i
\end{bmatrix}.$$
Further define $M'_i \in \F^{2 \times 2}$ by $$M' = \begin{bmatrix}
\omega_i & 1 \\
1 & 1
\end{bmatrix}.$$
$M'_i$ is a permutation of the rows and columns of $M_i$, so it suffices to prove the result for $\bigotimes_{i=1}^n M'_i$ instead of $\bigotimes_{i=1}^n M_i$. For $i \in [n]$, letting $g_i : \{0,1\} \to \F$ be defined by $g_i(0)=\omega_i$ and $g_i(1) =1$, we see that $M'_i = V_{g_i}$. Thus, defining $f : \{0,1\}^n \to \F$ by $$f(z[1], \ldots, z[n]) = \prod_{i=1}^{n} g(z[i]),$$ it follows that $\bigotimes_{i=1}^n M'_i = V_f$, as desired.
\end{proof}

\begin{lemma} \label{lem:prodtofunct}
For any field $\F$, positive integer $n$, and outer-nonzero matrices $M_1, \ldots M_n \in \F^{2 \times 2}$, there is a function $f : \{0,1\}^n \to \F$ and weighted permutation matrices $\Pi_n, \Pi'_n \in \F^{2^n \times 2^n}$ such that $$\bigotimes_{i=1}^n M_i = \Pi_n \times V_f \times \Pi'_n.$$
\end{lemma}

\begin{proof}
By \Cref{reducedim}, there are outer-1 matrices $M'_1, \ldots, M'_n \in \F^{2 \times 2}$ and invertible diagonal matrices $D, D' \in \F^{2^n \times 2^n}$ such that $\bigotimes_{i=1}^n M_i = D \times (\bigotimes_{i=1}^n M'_i) \times D'$. The result then follows by applying \Cref{lem:outer1tofunct} to $\bigotimes_{i=1}^n M'_i$.
\end{proof}

\begin{theorem} \label{rigthm2by2}
For any field $\F$ and positive integer $n$, let $M \in \F^{2^n \times 2^n}$ be a matrix of any of the following forms:
\begin{itemize}
    \item $M = V_f$ for any function $f : \{0,1\}^n \to \F$, or
    \item $M = \bigotimes_{\ell=1}^n M_i$ for any matrices $M_1, \ldots, M_n \in \F^{2 \times 2}$.
\end{itemize}
Then, for any $a \in (0,1)$, we have $\rig_{M}^{rc}\left(4 \cdot \binom{n}{<an}\right) \leq \binom{(1-a)n}{\leq (1-2a)n}^2$ over $\F$. In particular:
\begin{itemize}
    \item For sufficiently small $\eps>0$ we have $\rig_{M}^{rc}(2^{(1 - \Theta(\eps^2 / \log^2(1/\eps)) \cdot n}) \leq 2^{\eps \cdot n}$, and
    \item We have $\rig_{M}^{rc}(O(2^{0.981 \cdot n})) < o( 2^n )$.
\end{itemize}
\end{theorem}

\begin{proof}
For $M = V_f$, this follows by substituting the expression from \Cref{diagexpression} and the rigidity bound from \Cref{thm:rnotrigid} into \Cref{lem:productnotrigid}.

For $M = \bigotimes_{\ell=1}^n M_i$, let $k$ be the number of $i \in \{1,2,\ldots,n\}$ such that $M_i$ has at most two nonzero entries, and assume without loss of generality that $M_1, \ldots, M_k$ are the matrices with at most two nonzero entries. 

For each $i > k$ we can permute the rows and columns of the $2 \times 2$ matrix $M_i$ so that it is an outer-nonzero matrix, so combining \Cref{lem:prodtofunct} with \Cref{diagexpression} shows that we can write
$$\bigotimes_{\ell=k+1}^n M_i = \Pi_{n-k} \times R_{n-k} \times D \times R_{n-k} \times \Pi_{n-k}',$$
where $\Pi_{n-k}, D, \Pi_{n-k}' \in \F^{2^{n-k} \times 2^{n-k}}$ are weighted diagonal matrices, and $R_{n-k} \in \F^{2^{n-k} \times 2^{n-k}}$ is the disjointness matrix.

For each $i \leq k$, we can permute the rows and columns of the $2 \times 2$ matrix $M_i$ so that its nonzero entries are a subset of those of $R_1$. It follows that there is a matrix $B \in \F^{2^k \times 2^k}$ whose nonzero entries are a subset of those of $R_k$ such that $B = \bigotimes_{\ell=1}^k M_i$.

Letting $I_{2^k} \in \F^{2^k \times 2^k}$ denote the identity matrix, and applying \Cref{mixedproductproperty}, we can write
$$\bigotimes_{\ell=1}^n M_i = (\Pi_{n-k} \otimes I_{2^k}) \times (R_{n-k} \otimes B) \times (D \otimes I_{2^k}) \times (R_{n-k} \otimes I_{2^k}) \times (\Pi_{n-k}' \otimes I_{2^k}).$$

The three matrices $\Pi_{n-k} \otimes I_{2^k}$, $D \otimes I_{2^k}$, and $\Pi_{n-k}' \otimes I_{2^k}$ are weighted permutation matrices. The rigidity bound of \Cref{thm:rnotrigid} holds for the two matrices $R_{n-k} \otimes B$ and $R_{n-k} \otimes I_{2^k}$, since they are each Kronecker products of $R_{n-k}$ and a matrix whose nonzero entries are a subset of those of $R_k$ (after permuting the rows of $I_{2^k}$), and so their nonzero entries are a subset of those of $R_n$. We can thus once again apply \Cref{lem:productnotrigid} to conclude the desired rigidity upper bound for $M$.
\end{proof}

\section{Extension to Kronecker Products of Larger Matrices} \label{sec:bigproof}

\begin{theorem} \label{thm:mainrigidity}
For any field $\F$, positive integer $q>1$, matrices $M_1, \ldots, M_n \in \F^{q \times q}$, and sufficiently small $\eps>0$, the Kronecker product $M := \bigotimes_{\ell=1}^n M_\ell \in \F^{N \times N}$ for $N =  q^n$ has $$\rig_M^{rc}(N^{1 - O(2^{-q} q \log(q) \cdot \eps^2 / \log^2(1/\eps))}) \leq N^{\eps},$$
where the $O$ hides a universal constant. In particular, if $q \leq O(\log n)$, then $M$ is not Valiant-rigid.
\end{theorem}
In the remainder of this section, we prove \Cref{thm:mainrigidity}. We proceed by induction on $q$. The base case $q=2$ was given by \Cref{rigthm2by2}. Suppose $q>2$, and that the result is known already for $q-1$.

We may assume that $M_\ell \in \F^{q \times q}$ is an outer-nonzero matrix for all $\ell \in [n]$ since our proof below will only use the pattern of nonzero entries of the matrix, similar to the proof of \Cref{rigthm2by2}. By \Cref{reducedim}, we may further assume without loss of generality that $M_\ell \in \F^{q \times q}$ is an outer-1 matrix for all $\ell \in [n]$.
For nonnegative integers $i$, let $J_i \in \F^{q^i \times q^i}$ denote the $q^i \times q^i$ matrix whose entries are all $1$s. There are thus outer-0 matrices $A_1, \ldots, A_n \in \F^{q \times q}$ such that $M_\ell = J_1 + A_\ell$ for each $\ell \in [n]$.

For each subset $K \subseteq [n]$ let $A_K := \bigotimes_{\ell \in K} A_\ell$. This is the Kronecker product of $|K|$ different $(q-1)\times(q-1)$ matrices, padded with $(q^{|K|} - (q-1)^{|K|})$ rows and columns of $0$s. By the inductive hypothesis, for every $\eps>0$, setting $\eps' = O(2^{q-1} (q-1) \log(q-1) \cdot \eps^2 / \log^2(1/\eps))$  there are matrices $L_K, S_K \in \F^{q^{|K|} \times q^{|K|}}$ such that:
\begin{itemize}
    \item $A_K = L_K + S_K$,
    \item $\rank(L_K) \leq (q-1)^{|K|\cdot(1-\eps')}$, and
    \item for a given row $x \in [q]_0^{|K|}$ of $S_k$:
    \begin{itemize}
        \item If there is any $i \in [|K|]_0$ such that $x[i]=0$, then every entry of row $x$ of $S_K$ is $0$,
        \item Otherwise, there are at most $(q-1)^{|K| \cdot \eps}$ nonzero entries in row $x$ of $S_K$.
    \end{itemize}
    (and similar for a given column of $S_k$), and thus $\rank(S_k) \leq (q-1)^{|K|}$.
\end{itemize}

Now we can expand $M$:
\begin{align*}M 
&= \bigotimes_{\ell=1}^n M_\ell \\
&= \bigotimes_{\ell=1}^n (J_1 + A_\ell) \\
&= \sum_{K \subseteq [n]} A_K^{\otimes K} \otimes J_1^{\otimes [n]\setminus K} &(*)\\
&= \left( \sum_{K \subseteq [n]} L_K^{\otimes K} \otimes J_1^{\otimes [n]\setminus K} \right) + \left( \sum_{K \subseteq [n]} S_K^{\otimes K} \otimes J_1^{\otimes [n]\setminus K} \right)\\
\end{align*}

Let us first note that the first of these two matrices has low rank. Indeed, its rank can be bounded as
\begin{align*}
\rank\left( \sum_{K \subseteq [n]} L_K^{\otimes K} \otimes J_1^{\otimes [n]\setminus K} \right) 
&\leq \sum_{K \subseteq [n]} \rank\left( L_K^{\otimes K} \otimes J_1^{\otimes [n]\setminus K} \right) \\
&= \sum_{K \subseteq [n]} \rank\left( L_K\right) \cdot \rank\left( J_1^{\otimes (n-|K|)} \right) \\
&= \sum_{K \subseteq [n]} \rank\left( L_K\right)  \\
&\leq \sum_{K \subseteq [n]} (q-1)^{|K|\cdot(1-\eps')}  \\
&= \sum_{k=0}^n \binom{n}{k} \cdot (q-1)^{k\cdot(1-\eps')}  \\
&= \left(1 +  (q-1)^{(1-\eps')} \right)^n  \\
&= q^{n \cdot (1-\eps'')},  \\
\end{align*}

where $\eps''$ is given by $$\eps'' := \frac{\log(\frac{q}{(q-1)^{1-\eps'}+1})}{\log(q)} = \eps' \cdot\frac{(q-1) \log(q-1)}{q \log(q)} + O(\eps'^2). $$

It remains to show that the second matrix, $\sum_{K \subseteq [n]} S_K^{\otimes K} \otimes J_1^{\otimes [n]\setminus K}$, is not rigid. We partition it into three parts, for some $\delta>0$ to be determined, and letting $a := (q-1)/q$:
\begin{align*}
&\sum_{K \subseteq [n]} S_K^{\otimes K} \otimes J_1^{\otimes [n]\setminus K} 
\\ &= \left( \sum_{\overset{K \subseteq [n]}{|K| < (a - \delta)\cdot n}} S_K^{\otimes K} \otimes J_1^{\otimes [n]\setminus K} \right) + \left( \sum_{\overset{K \subseteq [n]}{|K| > (a + \delta)\cdot n}} S_K^{\otimes K} \otimes J_1^{\otimes [n]\setminus K} \right) + \left( \sum_{\overset{K \subseteq [n]}{(a + \delta)\cdot n \geq |K| \geq (a - \delta)\cdot n}} S_K^{\otimes K} \otimes J_1^{\otimes [n]\setminus K} \right)
\end{align*}
We will show that the first and second parts are low-rank, and that the third part is non-rigid. For the first, we bound similar to before (and using \Cref{binentropy} to bound $H$) that:
\begin{align*}
\rank\left( \sum_{\overset{K \subseteq [n]}{|K| < (a - \delta)\cdot n}} S_K^{\otimes K} \otimes J_1^{\otimes [n]\setminus K} \right)
&\leq \sum_{\overset{K \subseteq [n]}{|K| < (a - \delta)\cdot n}} \rank\left(  S_K^{\otimes K} \otimes J_1^{\otimes [n]\setminus K} \right) \\
&\leq \sum_{k=0}^{(a - \delta)\cdot n} \binom{n}{k} \cdot (q-1)^k \\
&\leq ((a - \delta)\cdot n) \cdot \binom{n}{(a - \delta)\cdot n} \cdot (q-1)^{(a - \delta)\cdot n} \\
&\leq O(n^2) \cdot 2^{H(a - \delta) \cdot n} \cdot (q-1)^{(a - \delta)\cdot n} \\
&= O(n^2) \cdot 2^{H(1/q + \delta) \cdot n} \cdot (q-1)^{(a - \delta)\cdot n} \\
&\leq 2^{(\log_2(q) - a \log_2(q-1) + \delta \log_2(q-1) - \Theta(q \cdot \delta^2)) \cdot n} \cdot (q-1)^{(a - \delta)\cdot n} \\
&= 2^{(\log_2(q) - \Theta(q \cdot \delta^2)) \cdot n} \\
&= q^{n(1 - \Theta(\delta^2 q / \log(q)))}. \\
\end{align*}

We can almost identically bound the rank of the second part by:
$$\rank\left( \sum_{\overset{K \subseteq [n]}{|K| > (a + \delta)\cdot n}} S_K^{\otimes K} \otimes J_1^{\otimes [n]\setminus K} \right) \leq O(n^2) \cdot 2^{H(1/2-\delta) \cdot n} \cdot (q-1)^{(a+\delta)\cdot n} \leq q^{n(1 - \Theta(\delta^2 q / \log(q)))}.$$

Finally, it remains to consider the third part:
$$B := \sum_{\overset{K \subseteq [n]}{(a + \delta)\cdot n \geq |K| \geq (a - \delta)\cdot n}} S_K^{\otimes K} \otimes J_1^{\otimes [n]\setminus K}.$$

We will show that after a small number of rows and columns of $B$ are removed, it is a sparse matrix. Since changing one row or column of a matrix is a rank-$1$ update, this will show that $B$ is not rigid and complete our proof.

The rows and columns we remove are those corresponding to $x \in [q]_0$ with $\nnz(x) \geq (a + \delta)\cdot n$. The number of these rows and columns is $$\sum_{k=(a + \delta)\cdot n}^{n} \binom{n}{k} \cdot (q-1)^{n-k},$$ which is again upper bounded by $q^{n(1 - \Theta(\delta^2 q / \log(q)))}$ similar to the previous two sums. 

Finally, let us show that there are not many nonzero entries remaining in any row or column of $B$. Consider a row $x \in [q]_0$ that we did not remove, meaning $\nnz(x) < (a + \delta)\cdot n$. Suppose, for some $K \subseteq [n]$ with $(a + \delta)\cdot n \geq |K| \geq (a - \delta)\cdot n$, that $S_K^{\otimes K} \otimes J_1^{\otimes [n]\setminus K}$ has nonzero entries in row $x$. That means there cannot be any $\ell \in K$ such that $x[\ell]=0$. The number of choices for $K$ is hence at most
$$\sum_{k=(a - \delta)\cdot n}^{(a + \delta)\cdot n}\binom{\nnz(x)}{k} \leq (2 \delta n) \cdot \binom{(a+\delta) \cdot n}{(a - \delta)\cdot n} \leq O(n) \cdot 2^{(a+\delta)\cdot H(2\delta/(a+\delta))\cdot n} \leq 2^{\Theta(\delta \cdot \log(1/\delta)) \cdot n}.$$

For each such $K$, how many nonzero entries does it contribute to row $x$? A simple upper bound is $\nnzr(S_K) \cdot \nnzr(J_1^{\otimes(n-|K|)})$, but we can get a better bound by noting that many of the columns with those nonzero entries have been removed. Indeed, for a $y \in [q]_0^n$, the entry $B[x,y]$ will be nonzero and not removed earlier only if:
\begin{itemize}
    \item $\nnz(y) < (a + \delta) \cdot n$, and
    \item $S_K[x|_K,y|_K] \neq 0$.
\end{itemize}
In particular, this latter condition requires that $\nnz(y|_K) = |K|$, which means only $(a + \delta) \cdot n - |K| \leq 2 \delta n$ entries of $y|_{[n]\setminus K}$ may be nonzero. There are thus:
\begin{itemize}
    \item $\leq (q-1)^{|K| \cdot \eps}$ choices for $y|_K$, by definition of $S_K$, and
    \item $\leq \binom{n-|K|}{2\delta n} \cdot (q-1)^{2 \delta n}$ choices for $y|_{[n]\setminus K}$ because at most $2\delta n$ of its entries may be nonzero.
\end{itemize}
The total number of such $y$ is thus at most
\begin{align*}
    (q-1)^{|K| \cdot \eps} \cdot \binom{n-|K|}{2\delta n} \cdot (q-1)^{2 \delta n}
    &\leq (q-1)^{(a+\delta) \cdot n \cdot \eps} \cdot \binom{(1/q + \delta)n}{2\delta n} \cdot (q-1)^{2 \delta n} \\
    &\leq O(n) \cdot 2^{n \cdot (\eps(a+\delta)\log(q-1) + (1/q+\delta)H(2\delta/(1/q+\delta)) + 2\delta \log(q-1))} \\
    &\leq O(n) \cdot 2^{n \cdot ((a\eps + 2\delta + \eps\delta)\log(q-1) + 2\delta \log((1/q+\delta)/(2\delta)))} \\
    &\leq 2^{n \cdot (a\eps\log(q-1) + 2\delta\log(1/\delta) + O(\delta))} \\
    &= q^{n \cdot (\frac{(q-1)\log(q-1)}{q \log(q)}\eps + 2\delta\log(1/\delta)/\log(q) + O(\delta))}.
\end{align*}

In summary, $M$ can be written as the sum of a matrix of rank at most
$$q^{n\cdot(1- \frac{(q-1)\log(q-1)}{q \log q} \eps' + O(\eps'^2))} + q^{n\cdot(1 - \Theta(\delta^2 q / \log(q)))},$$
and a matrix with row/column sparsity at most
$$q^{n \cdot (\frac{(q-1)\log(q-1)}{q \log(q)}\eps + 2\delta\log(1/\delta)/\log(q) + O(\delta))}.$$

Let $c = \frac{(q-1)\log(q-1)}{2 q \log(q)}\eps$, so that $\eps' = O(2^{q} (q-1) \log(q-1) \cdot \eps^2 / \log^2(1/\eps)) = O(\frac{2^{q-1} q^2 \log^2(q)}{(q-1)\log(q-1)} \cdot c^2 / \log^2(1/c))$, and pick $\delta$ such that $c = \delta \log(1/\delta) / \log(q)$. This shows as desired that $$\rig_M^{rc}(N^{1 - O(2^q q \log(q) \cdot c^2 / \log^2(1/c))}) \leq N^{c}.$$

\subsection{Extension to Functions with Larger Domains}

\begin{theorem} \label{thm:mainrigidity2}
For any field $\F$, positive integer $q>1$, and function $f : \{0,1,\ldots,q-1\}^n \to \F$, define the matrix $V_f \in \F^{q^n \times q^n}$ by, for $x,y \in \{0,1,\ldots,q-1\}^n$, $$V_f[x,y] = f\left( \max\{x[0],y[0]\}, \max\{x[1],y[1]\}, \max\{x[2],y[2]\}, \ldots, \max\{x[n-1],y[n-1]\} \right).$$
For any sufficiently small $\eps>0$, the matrix $V_f \in \F^{N \times N}$ for $N =  q^n$ has $$\rig_{V_f}^{rc}(N^{1 - O(2^{-q} q \log(q) \cdot \eps^2 / \log^2(1/\eps))}) \leq N^{\eps},$$
where the $O$ hides a universal constant. In particular, if $q \leq O(\log n)$, then $V_f$ is not Valiant-rigid.
\end{theorem}

\begin{proof}
Just like in the proof of \Cref{thm:mainrigidity}, we proceed by induction on $q$. The base case $q=2$ was given by \Cref{rigthm2by2}. Suppose $q>2$, and that the result is known already for $q-1$.

For any $T \subseteq [n]_0$, we define $g_T : [q-1]_0^{|T|} \to [q]_0^n$ as follows. Let $t_1, t_2, \ldots, t_{|T|}$ be an enumeration of the elements of $T$. Then, for $z \in [q-1]_0^{|T|}$ and $i \in [n]_0$ we define:
$$g_T(z)[i] := \begin{cases}
    0 &\text{ if } i \notin T,\\
    z_{t_j} + 1 &\text{ if } i = t_j \in T.
\end{cases}$$

For every set $S \subseteq [n]$, we define the function $f_S : \{0,1,2,\ldots,q-2\}^{|S|} \to \F$ as, for any $z \in [q]_0^n$,
$$f_S(z) = \sum_{T \subseteq S} (-1)^{|S|-|T|} \cdot f(g_T(z)).$$

I now claim that
$$V_f = \sum_{S \subseteq [n]} V_{f_S}^{\otimes S} \otimes J_1^{\otimes [n] \setminus S}.$$
Once I show this, we can simply substitute it in for Equation~$(*)$ in the proof of \Cref{thm:mainrigidity}, and the remainder of the proof is exactly the same (with $A_K$ replaced by $V_{f_K}$ throughout). 

For $z \in [q]_0^n$, let $S_z \subseteq [n]$ be the set of indices $i$ with $z[i] \neq 0$. Notice that, for $x,y \in [q]_0^n$, letting $z \in [q]_0^n$ be the entry-wise max of $x$ and $y$, we have that:
$$\left( \sum_{S \subseteq [n]} V_{f_S}^{\otimes S} \otimes J_1^{\otimes [n] \setminus S} \right)[x,y] =  \sum_{S \subseteq [n]} \left(V_{f_S}^{\otimes S} \otimes J_1^{\otimes [n] \setminus S} [x,y]\right) =  \sum_{S \subseteq [n]} ([S \subseteq S_z] ~?~ f_S(z) : 0).$$

 It thus suffices to show that for all $z \in [q]_0^n$, we have $\sum_{S \subseteq  S_z} f_S(z) = f(z)$. We can verify this by using inclusion-exclusion:
\begin{align*}\sum_{S \subseteq  S_z} f_S(z) &= \sum_{S \subseteq  S_z} \sum_{T \subseteq S} (-1)^{|S|-|T|} \cdot f(g_T(z))  
\\ &= \sum_{T \subseteq S_z} \sum_{T \subseteq S \subseteq S_z}(-1)^{|S|-|T|} \cdot f(g_T(z))
\\ &= \sum_{T \subseteq S_z} f(g_T(z)) \cdot \sum_{T \subseteq S \subseteq S_z}(-1)^{|S|-|T|} 
\\ &= \sum_{T \subseteq S_z} f(g_T(z)) \cdot \sum_{k=0}^{|S_z|-|T|} \binom{|S_z| - |T|}{k} \cdot (-1)^k
\\ &= f(g_{S_z}(z))
\\ &= f(z).
\end{align*}
Here, we used the fact that $\sum_{k=0}^n \binom{n}{k} \cdot (-1)^k = 0$ unless $n=0$.
\end{proof}

Note that \Cref{thm:mainrigidity2} also holds with `$\max$' replaced with `$\min$', as this corresponds to appropriately permuting the truth table of $f$.

\section{Kronecker Products and Matrix Multiplication} \label{sec:mmult}

\begin{definition}
For any field $\F$ and positive integers $m,n,p$, let $MM_\F(m,n,p)$ denote the smallest size of an arithmetic circuit for computing the product of an $m \times n$ matrix and a $n \times p$ matrix over $\F$. For instance, $MM_\F(n,n,n) \leq n^{\omega + o(1)}$ where $\omega \leq 2.373$~\cite{williams2012multiplying,le2014powers} is the matrix multiplication exponent.
\end{definition}

\begin{lemma} \label{lem:prodidentity}
For any field $\F$, positive integers $q,N,$ and matrix $M \in \F^{q \times q}$, the linear transformation $M \otimes I_{N}$ can be computed by an arithmetic circuit of size $MM_\F(q,q,N)$.
\end{lemma}

\begin{proof}
Computing $(M \otimes I_{N}) \times v$ for a vector $v \in \F^{q \cdot N}$ is equivalent to computing $M \times v_\ell$ for all $N$ of the vectors $v_1, \ldots, v_N \in \F^q$ whose concatenation gives $v$. This, in turn, is equivalent to multiplying $M \times (v_1 | v_2 | \cdots | v_N)$, which can be done with a circuit of size $MM_\F(q,q,N)$ as desired.
\end{proof}

\begin{lemma} \label{butterflytomm}
For any field $\F$, positive integers $q,n,k$ such that $k$ divides $n$, and matrices $M_1, \ldots, M_n \in \F^{q \times q}$, the linear transformation $M := \bigotimes_{\ell=1}^n M_\ell \in \F^{q^n \times q^n}$ can be computed by an arithmetic circuit of size $k \cdot MM_\F(q^{n/k},q^{n/k},q^{n \cdot (k-1)/k})$.
\end{lemma}

\begin{proof}
For each $\ell \in [k]_0$, define the matrix $M'_\ell \in \F^{q^{n/k} \times q^{n/k}}$ by $$M'_\ell := \bigotimes_{i=1}^{n/k}M_{i + \ell \cdot n/k}.$$ Hence, $$\bigotimes_{\ell=0}^{k-1} M'_\ell = \bigotimes_{i=1}^n M_i = M_n.$$ Applying \Cref{butterfly} to the $M'_\ell$ matrices shows that, in order to compute $M_n$, it suffices to compute $k$ linear transformations, where the $\ell$th, for $\ell \in [k]_0$, is a permutation of the rows and columns of $M'_\ell \otimes I_{q^{n \cdot (k-1)/k}}$. By \Cref{lem:prodidentity}, each can be computed by an arithmetic circuit of size $MM_\F(q^{n/k},q^{n/k},q^{n \cdot (k-1)/k})$, as desired.
\end{proof}

\begin{corollary} \label{cor:mmtohad}
Suppose that, for any integer $k>1$, we have $MM_\F(n,n,n^{k-1}) \leq o(n^k \log n)$. Then, for any field $\F$, fixed positive integer $q$, positive integer $n$, and matrices $M_1, \ldots, M_n \in \F^{q \times q}$, the linear transformation $M := \bigotimes_{\ell=1}^n M_\ell \in \F^{N \times N}$ (with $N = q^n$) can be computed by an arithmetic circuit of size $o(N \log N)$. 
\end{corollary}

\begin{proof}
Applying \Cref{butterflytomm}, we see that $M$ can be computed by an arithmetic circuit of size $k \cdot MM_\F(q^{n/k},q^{n/k},q^{n \cdot (k-1)/k})$. By assumption, this is $o((q^{n/k})^k \log(q^{n/k})) = o(N \log N)$, as desired.
\end{proof}

In fact, as $k$ gets large, it is known that the exponent of $MM_\F(n,n,n^{k-1})$ is the desired $k$:

\begin{proposition}[\cite{huang1998fast}] \label{prop:rectmm}
For every field $\F$ and integer $k>1$, we have $MM_\F(n,n,n^{k-1}) \leq O(n^{k \cdot \log_{k-1}(k)})$. Here, the $O$ is hiding a function of $k$. Note that the exponent is $$k \cdot \log_{k-1}(k) = k + O\left(\frac{1}{\log k}\right).$$
\end{proposition}

\begin{proof}[Proof sketch]
This follows from \cite[{Equation~(7.1)}]{huang1998fast}. In the notation of their Equation~(7.1), using $q=r=k$ and a small $\beta>0$, we find that $\omega(1,1,k) < (k+1) \cdot \log_k(k+1)$. The result then follows by applying Sch{\"o}nhage's theorem~\cite{schonhage1981partial}, using the notation of \cite[{Theorem~2.1}]{huang1998fast} with $\eps = (k+1) \cdot \log_k(k+1) - \omega(1,1,k)$, which is a function of only $k$.
\end{proof}

Unfortunately, in order to combine \Cref{prop:rectmm} with \Cref{cor:mmtohad} to construct an arithmetic circuit of size $o(N \log N)$, we would need to pick $k = \Omega(\log N / \log\log N)$ in order for the non-leading term from $MM_\F(n,n,n^{k-1})$ (i.e. $(N^{1/k})^{O(1/\log k)} = N^{O(1/k\log k)}$) to be negligible. However, in that case, the $O$ in \Cref{prop:rectmm} is hiding a growing function of $N$, which swamps our savings unless that growing function is relatively small:

\begin{corollary}
Let $f(k)$ be the constant factor hidden in \Cref{prop:rectmm}, and suppose that $f(k) < o(\log k)$. Then, for any field $\F$, fixed positive integer $q$, positive integer $n$, and matrices $M_1, \ldots, M_n \in \F^{q \times q}$, the linear transformation $M := \bigotimes_{\ell=1}^n M_\ell \in \F^{N \times N}$ (with $N = q^n$) can be computed by an arithmetic circuit of size $o(N \log N)$. 
\end{corollary}

\begin{proof}
Applying \Cref{butterflytomm} with $k = \log N / \log\log N$, the resulting circuit size upper bound is $O(k \cdot f(k) \cdot N) < o(k \log k \cdot N) = o(N \log N)$. 
\end{proof}

\section{Arithmetic Complexity}

In this section, we focus on the complexity of linear transformations using arithmetic circuits in which each gate has fan-in 2. This is often the best model for counting the exact number of arithmetic operations needed to compute a given linear transformation.

\begin{lemma}
For any field $\F$ and positive integer $n$, let $M \in \F^{2^n \times 2^n}$ be a matrix of any of the following forms:
\begin{itemize}
    \item $M = V_f$ for any function $f : \{0,1\}^n \to \F$, or
    \item $M = \bigotimes_{\ell=1}^n M_i$ for any matrices $M_1, \ldots, M_n \in \F^{2 \times 2}$.
\end{itemize}
Then, $M^{\otimes n} \in \F^{N \times N}$ (with $N = 2^n$) can be computed by an arithmetic circuit with $N \log_2 N$ addition gates and $3N$ multiplication gates.
\end{lemma}

\begin{proof}
By \Cref{diagexpression} and \Cref{lem:prodtofunct}, any such $M$ can be written as the product of three diagonal matrices and two copies of $R_n$. It thus suffices to show that $R_n$ has an arithmetic circuit with $\frac12 N \log_2 N$ addition gates. By \Cref{butterfly}, to compute $R_n$, it suffices to compute $\log_2 N$ different copies of $A := R_1 \otimes I_{N/2}$. In $A$, half the rows have two $1$s, which can be computed by a single addition gate, and the other half of the rows have a single $1$ and don't need any gates to compute (we just output one of the inputs). Thus, in total, $A$ needs $N/2$ addition gates, so $R_n$ needs $\frac12 N \log_2 N$ addition gates, as desired.
\end{proof}

In fact, we can make this algorithm uniform, since the relevant diagonal matrices can all also be constructed by evaluating $R_n$:

\begin{lemma}
For any field $\F$, positive integer $n$, and function $f : \{0,1\}^n \to \F$, letting $N = 2^n$, suppose there is an algorithm that outputs the truth table of $f$ (i.e. evaluates $f$ on all $N$ inputs from $\{0,1\}^n$) in time $T$. Let $M$ be the time to perform a multiplication over $\F$, and $A$ be the time to perform an addition or subtraction over $\F$. Then, there is an algorithm which, given as input $x \in \F^{N}$, outputs $V_f \times x$ in time $O(T + A \cdot N \log N + M \cdot N)$.
\end{lemma}

For $f = AND$, this corresponds to the algorithm for the Orthogonal Vectors problem with $n$ vectors in dimension $d$ with running time $O(n + d \cdot 2^d)$. We hence get a the same running time for any such problem for a function $f : \{0,1\}^n \to \F$.

\section{Generalizing the Approach of \Cref{sec:framework}} \label{sec:moregeneral}

In \Cref{sec:framework} we showed how to convert a rigidity upper bound for a matrix $M$ into a low-depth circuit upper bound for $M^{\otimes n}$. A key intermediate step was that from a circuit upper bound for $M$ itself, one can take Kronecker powers to get a circuit for $M^{\otimes n}$ for any $n$. In this section, we generalize this to show that if $M \in \F^{q \times q}$ has a nontrivial construction $M = B_1 \times B_2 \times \cdots B_d$ where $\prod_{i=1}^d \nnz(B_i) < q^{d+1}$ then this can still give a nontrivial circuit upper bound for $M^{\otimes n}$ of depth $d$ and size $O(q^{n(1 + (1 - \eps)/d)})$, even if $\nnz(B_i)$ is greater than $q^{1 + 1/d}$ for some of the $i$. Note that we can achieve $\prod_{i=1}^d \nnz(B_i) = q^{d+1}$ by picking $B_1 = M$ and $B_2 = \cdots = B_d = I_q$. This more general result was not needed in our construction in \Cref{sec:framework}, since the constructions from non-rigidity were naturally symmetric, but they could be useful for designing upper bounds in other ways.

\begin{lemma} \label{lem:evenout}
For any field $\F$ and positive integers $q,d$, and matrix $M \in \F^{q \times q}$, suppose there are real numbers $a_1, \ldots, a_d \geq 1$ such that, for any positive integer $n$, the matrix $M^{\otimes n}$ can be written as $M^{\otimes n} = A_{n,1} \times A_{n,2} \times \cdots \times A_{n,d}$ for some matrices with $\nnz(A_{n,\ell})=O(q^{a_\ell \cdot n})$ for all $\ell \in [d]$. Let $j^* = \argmax_{j \in [d]} a_j$, and let $$a := 1 + \frac{a_{j^*} - 1}{1 + d \cdot a_{j^*} - \sum_{j=1}^d a_j}.$$ Then, for any positive integer $n$, we can write $M^{\otimes n} = B_{n,1} \times B_{n,2} \times \cdots \times B_{n,d}$ for some matrices with $\nnz(B_{n,j})=O(q^{a \cdot n})$ for all $j \in [d]$.

In particular, if $(\sum_{j=1}^d a_j)/d < 1 + 1/d$, then $a < 1 + \frac{1}{d}$.
\end{lemma}

\begin{proof}
We first need one piece of notation: For matrices $S,T$ of the same dimensions, and a Boolean predicate $P$, we write $(P ~?~ S : T)$ to denote the matrix $$(P ~?~ S : T) := \begin{cases} S &\text{ if $P$ is true,} \\ T &\text{ if $P$ is false.} \end{cases}$$

Let $b, b_1, \ldots, b_d$ be positive real numbers which sum to $1$ to be determined. By assumption, for each $j \in [d]$, there is a matrix $A_{bn, j}$ with $\nnz(A_{bn, j}) = O(q^{b \cdot a_j \cdot n})$, and $M^{\otimes b n} = \prod_{j=1}^d A_{bn,j}$. We can hence write:
\begin{align*}
    M^{\otimes n} &= M^{\otimes b n} \otimes \bigotimes_{\ell=1}^d M^{\otimes b_\ell \cdot n} \\
    &= \left( \prod_{j=1}^d A_{bn,j} \right) \otimes \bigotimes_{\ell=1}^d \left( \prod_{j=1}^d ([j=\ell] ~?~ M^{\otimes b_\ell \cdot n} : I_{q^{b_\ell \cdot n}}) \right) \\
    &= \prod_{j=1}^d \left( A_{bn,j} \otimes \bigotimes_{\ell=1}^d ([j=\ell] ~?~ M^{\otimes b_\ell \cdot n} : I_{q^{b_\ell \cdot n}}) \right)\\
    &= \prod_{j=1}^d P_j \times \left( A_{bn,j} \otimes M^{\otimes b_j n} \otimes I_{q^{n(1-b-b_j)}} \right) \times P'_j,\\
\end{align*}
for appropriate permutation matrices $P_j, P'_j$ for each $j \in [d]$, by \Cref{mixedproductproperty}. We will pick $$B_{n,j} := P_j \times \left( A_{bn,j} \otimes M^{\otimes b_j n} \otimes I_{q^{n(1-b-b_j)}} \right) \times P'_j,$$ so it is indeed the case that $M^{\otimes n} = B_{n,1} \times B_{n,2} \times \cdots \times B_{n,d}$. Let us now bound $\nnz(B_{n,j})$:
\begin{align*}
    \nnz(B_{n,j}) &= \nnz(P_j \times \left( A_{bn,j} \otimes M^{\otimes b_j n} \otimes I_{q^{n(1-b-b_j)}} \right) \times P'_j) \\
    &= \nnz(A_{bn,j} \otimes M^{\otimes b_j n} \otimes I_{q^{n(1-b-b_j)}} ) \\
    &= \nnz(A_{bn,j}) \cdot \nnz( M^{\otimes b_j n} ) \cdot \nnz( I_{q^{n(1-b-b_j)}} ) \\
    &\leq O(q^{b \cdot a_j \cdot n}) \cdot q^{2b_j n}  \cdot q^{n(1-b-b_j)}  \\
    &= O(q^{(1 + b_j + (a_j-1)b) \cdot n} ).
\end{align*}

We pick $$b := \frac{1}{1 + \sum_{j=1}^d (a_{j^*} - a_j)},$$ and for all $j \in [d]$, we pick $$b_j := (a_{j^*} - a_j) \cdot b,$$
so that $b + \sum_{j=1}^d b_j = 1$. Hence, for every $j \in [d]$, we have from the calculation above that
$$\nnz(B_{n,j}) \leq O(q^{(1 + b_j + (a_j-1)b) \cdot n} ) = O(q^{(1 + (a_{j^*} - a_j) \cdot b + (a_j-1)b) \cdot n} ) = O(q^{(1 + (a_{j^*}-1)b) \cdot n} ),$$
as desired.

For the `in particular' sentence of the Lemma statement: Suppose $\sum_{j=1}^d a_j / d = 1 + c/d$ for some $0 \leq c < 1$. It follows that $$a = 1 + \frac{a_{j^*} - 1}{1 + d \cdot a_{j^*} - d - c}.$$
The derivative of this expression with respect to $a_{j^*}$ is $(1-c)/(a_{j^*} d - c - d + 1)^2$, which is always nonnegative, so for a fixed $c$, the value of $a$ is maximized when $a_{j*}$ is as large as possible. Since $a_j \geq 1$ for all $j \in [d]$, we must have that
$$a_{j^*} = \left(\sum_{j=1}^d a_j \right) - \left(\sum_{j \in [d], j \neq j^*} a_j \right) \leq (d+c) - (d-1) \cdot 1 = c+1.$$
We therefore have that
$$a \leq  1 + \frac{(c+1) - 1}{1 + d \cdot (c+1) - d - c} = 1 + \frac{c}{1+c(d-1)} < 1 + \frac{1}{d},$$
as desired.
\end{proof}

When the matrix $M$ is symmetric (i.e. satisfies $M = M^T$), we can get an improved exponent (by improving on the choice of $a_{j^*}$):
\begin{lemma}\label{lem:evenoutsym}
For any field $\F$ and positive integers $q,d$, and matrix $M \in \F^{q \times q}$ with $M=M^T$, suppose there are real numbers $a_1, \ldots, a_d \geq 1$ such that, for any positive integer $n$, the matrix $M^{\otimes n}$ can be written as $M^{\otimes n} = A_{n,1} \times A_{n,2} \times \cdots \times A_{n,d}$ for some matrices with $\nnz(A_{n,\ell})=O(q^{a_\ell \cdot n})$ for all $\ell \in [d]$. Define $a_{j^*} := \max_{j \in [d]} (a_j+ a_{d-j})/2$, and let $$a := 1 + \frac{a_{j^*} - 1}{1 + d \cdot a_{j^*} - \sum_{j=1}^d a_j}.$$ Then, for any positive integer $n$, we can write $M^{\otimes n} = B_{n,1} \times B_{n,2} \times \cdots \times B_{n,d}$ for some matrices with $\nnz(B_{n,j})=O(q^{a \cdot n})$ for all $j \in [d]$.

In particular, if $(\sum_{j=1}^d a_j)/d < 1 + 1/d$, then $a < 1 + \frac{1}{d}$.
\end{lemma}

\begin{proof}
We can write \begin{align*}
    M^{\otimes n} &= M^{\otimes n/2} \otimes (M^{\otimes n/2})^T 
    \\ &= \left( \prod_{j=1}^d A_{n/2,j} \right) \otimes \left( \prod_{j=1}^d A_{n/2,d-j}^T \right) 
    \\ &= \prod_{j=1}^d \left( A_{n/2,j}  \otimes  A_{n/2,d-j}^T \right).
\end{align*}
The result then follows by applying \Cref{lem:evenout} to this new expression of $M^{\otimes n}$ as a product of $d$ matrices, since for $\ell \in [d]$, we have
$$\nnz\left( A_{n/2,j}  \otimes  A_{n/2,d-j}^T \right) = \nnz\left( A_{n/2,j}\right) \cdot \nnz\left( A_{n/2,d-j}\right) \leq O\left( q^{(a_j + a_{d-j}) \cdot n/2} \right).$$
\end{proof}

\section*{Acknowledgements}
I would like to thank Amol Aggarwal, Chi-Ning Chou, Ben Edelman, Alexander Golovnev, DD Liu, Jon Schneider, Leslie Valiant, Virginia Vassilevska Williams, and Ryan Williams for helpful discussions throughout this project. I'd especially like to thank Virginia Vassilevska Williams for pointing out \Cref{prop:rectmm} to me, and anonymous reviewers for many helpful comments.

\bibliographystyle{alpha}
\bibliography{papers}

\end{document}